\documentclass[oneside, reqno]{amsart}
\usepackage{graphicx}
\usepackage{caption, subcaption, float}
\usepackage{tabularx}
\usepackage{xcolor}
\usepackage[margin=1in]{geometry}
\usepackage{amsmath}
\allowdisplaybreaks

\newcommand{\bea}{\begin{eqnarray}}
\newcommand{\eea}{\end{eqnarray}}
\newcommand{\beann}{\begin{eqnarray*}}
\newcommand{\eeann}{\end{eqnarray*}}

\newcolumntype{Y}{>{\centering\arraybackslash}X}

\usepackage{verbatim}
\usepackage{amssymb,amsmath,amsthm}
\usepackage[mathscr]{eucal}
\usepackage{tikz}
\usetikzlibrary{positioning}
\usepackage{pdfpages}
\usepackage{mathtools}
\usepackage{setspace}

\numberwithin{equation}{section}
\usepackage[toc,page]{appendix}

\usepackage{graphicx}
\usepackage{psfrag}
\usepackage{afterpage}
 \usepackage{array,multirow,graphicx}
\newtheorem{theorem}{Theorem}[section]

\newtheorem{proposition}[theorem]{Proposition}
\newtheorem{lemma}[theorem]{Lemma}
\newtheorem{cor}[theorem]{Corollary}
\newtheorem{rem}[theorem]{Remark}

\usepackage[colorlinks=true]{hyperref}
\usepackage{ifpdf}
\usepackage{float}

\title{Multiple Scale Asymptotics of Map Enumeration}

\author
{Nicholas Ercolani$^1$}
\thanks{$^1$ University of Arizona, Department of Mathematics
  (ercolani@math.arizona.edu)}

\author{Joceline Lega$^2$}
\thanks{$^2$ University of Arizona, Department of Mathematics (lega@math.arizona.edu, \url{www.math.arizona.edu/\string~lega/}),}

\author{Brandon Tippings$^3$}
\thanks{$^3$ University of Arizona, Department of Mathematics ({tippings@arizona.edu}).}

\begin{document}
\maketitle

\begin{abstract}
    We introduce a systematic approach to express generating functions for the enumeration of maps on surfaces of high genus in terms of a single generating function relevant to planar surfaces. Central to this work is the comparison of two asymptotic expansions obtained from two different fields of mathematics: the Riemann-Hilbert analysis of orthogonal polynomials and the theory of discrete dynamical systems. By equating the coefficients of these expansions in a common region of uniform validity in their parameters, we recover known results and provide new expressions for generating functions associated with graphical enumeration on surfaces of genera 0 through 7. Although the body of the article focuses on 4-valent maps, the methodology presented here extends to regular maps of arbitrary even valence and to some cases of odd valence, as detailed in the appendices.
\end{abstract}
%%%%%%%%%%%%%%%%%%%%%%%%%%%%%%%%%%%%%
%%%%%%%%%%%%Introduction%%%%%%%%%%%%%%%%%%

\section{Introduction}
\label{sec:mapsgraphs}

This paper combines ideas from random matrix theory and dynamical systems to address a long-standing question relevant to a particular branch of graph theory, specifically the enumeration of \textit{maps}. This branch of graphical enumeration arose in the mid-twentieth century as a first step in addressing the following general question: given a spatial graph, when can that graph be embedded on a particular type of topological surface? Some graphs are planar, meaning the graph can be embedded in a plane (or equivalently a sphere) without being forced to cross itself. The same question can be posed for more general surfaces, thereby setting up a kind of complexity classification of spatial graphs, or networks, in terms of the topology of surfaces on which they can or cannot be embedded. 

Being able to enumerate graphs subject to topological complexity serves as a first step in understanding the general role of topological frustration in network theory. There have been quite a few studies in the physics and mathematics literature related to this problem and in particular toward the construction of generating functions for this enumeration indexed by graph size (the number of vertices, which we will denote $j$). Because the graph size is not bounded, this potentially involves an infinite amount of information for each topological surface. However, it was shown in \cite{bib:er} that these generating functions depend only on a minimal, specific, {\it finite} set of rational parameters. The results discussed in this paper develop a systematic method for identifying these parameters explicitly.

A \textit{map} is a connected graph $\Gamma$ embedded in a surface $M$ that satisfies certain additional conditions. The surfaces we consider are compact, oriented and connected topological surfaces, each of them being uniquely specified, up to a homeomorphism, by its genus, $g$. Embedding a graph, $\Gamma$, into $M$ amounts to embedding its vertices and edges in such a way that the overall placement of the graph on $M$ is injective and continuous. The last additional condition required is that after the surface is cut along the edges of the embedded graph, what remains is a disjoint union of contractible topological cells. For fixed genus $g$, we refer to maps satisfying these conditions as $g$-\textit{maps}.

A depiction of a map in a {\it local chart} on a surface is illustrated by the dashed black graph embedded in a planar region shown in Figure \ref{dualmap}. Note that in this example all (black) vertices have valence 4 (in the graph-theoretic sense). Maps whose vertices all have the same valence, $\mathcal V$, are referred to as $\mathcal V$-regular maps in analogy with the terminology for graphs. Figure \ref{dualmap} also (locally) illustrates 
the dual map (depicted in terms of the solid blue graph). The 4-regularity of the original map results in the dual map being a tiling of the surface by topological rectangles. 
\begin{figure}[h] 
\begin{center} 
\includegraphics[width=.6\linewidth]{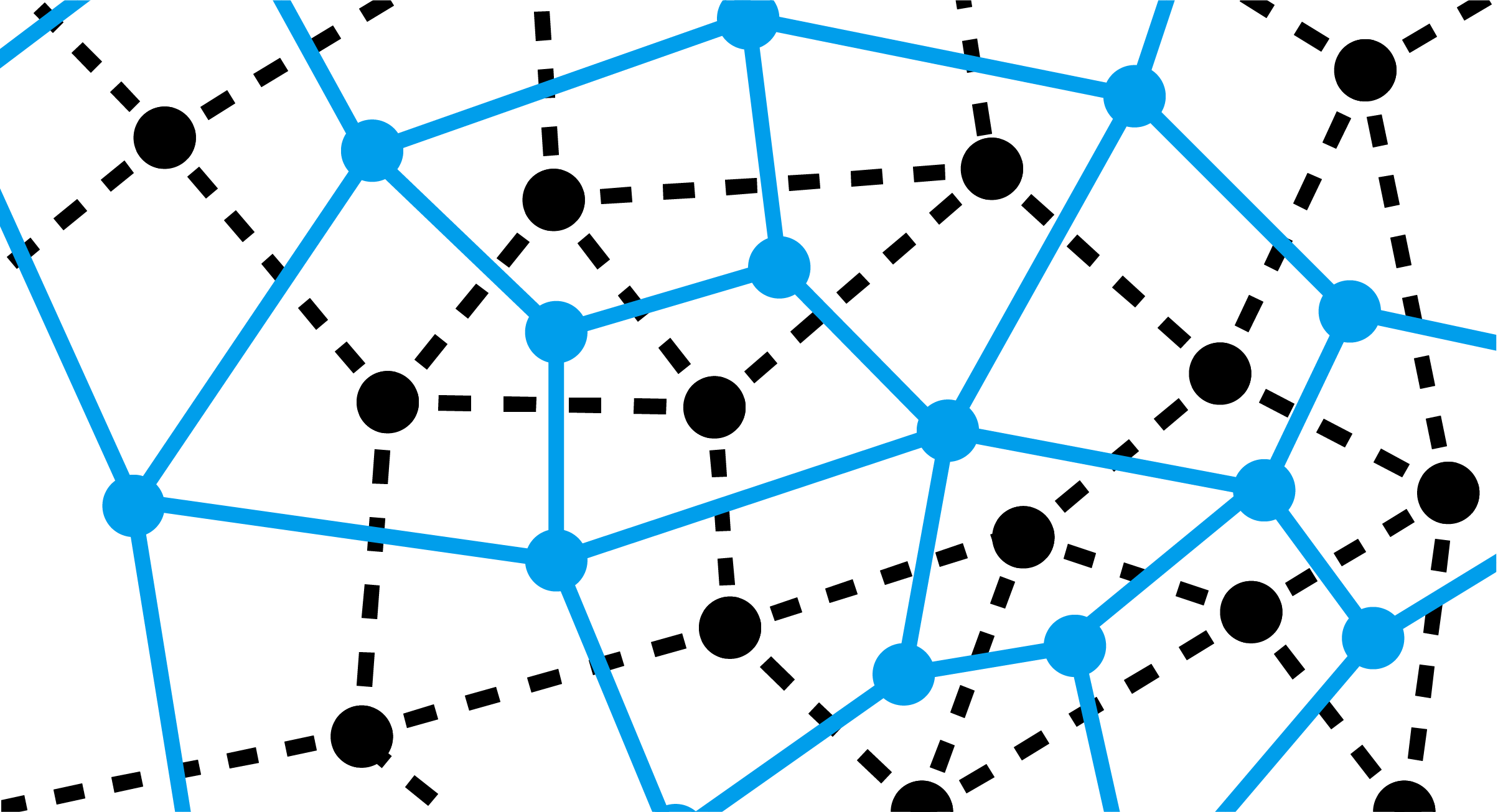}
\caption{Illustration of a 4-valent map in a local chart and of its dual.}  \label{dualmap}
\end{center}
\end{figure}

Such surface tilings arise in a number of settings where one may be interested in modelling some kind of large scale cellular growth subject to global topological constraints. Physical applications arise in pattern formation in foams \cite{bib:ball}, planar systems of interacting particles \cite{bib:l11}, embryo gastrulation \cite{bib:mb}, and vertex dynamics \cite{bib:fomg}. For related statistical or stochastic questions (such as statistical mechanics/dynamics on random networks \cite{bib:hs} or stochastic Loewner evolution of interfaces \cite{bib:cn}) the large scale enumeration of maps with fixed features  is an important initial problem.

As mentioned earlier, we are interested in the enumeration of maps with a fixed number, $j$, of vertices as $j$ varies and becomes large. To reduce such enumerations to a combinatorial question, one needs to define when two maps are equivalent. One counts maps modulo equivalence and the set of equivalence classes is finite. On a genus $g$ surface, two maps are equivalent if there is an orientation-preserving homeomorphism from the surface to itself that induces a homeomorphism of the graph to itself preserving the sets of vertices and edges but possibly respectively permuting them (while still preserving the incidence relations) \cite{bib:lz}. Equivalences for which such a permutation is non-trivial can arise. To avoid such technicalities at the outset, it is typical to consider the enumeration of {\it labelled} maps. These are maps in which the vertices are labelled (or numbered) and the edges around each vertex on the surface are also labelled consistent with the orientation of the surface. For the latter it suffices to label one initial edge. The orientation (say clockwise) will then order the successive numbering of the remaining edges around that vertex.  This labelling breaks any symmetries that could yield a non-trivial automorphism of $\Gamma$.
\medskip

The earliest work on map enumeration goes back to Tutte \cite{bib:tu}, using a purely combinatorial approach. Further results in this vein have continued up to the present time, producing some remarkable combinatorial insights \cite{bib:lz, bib:jv, bib:bgr, bib:cms, bib:chap}. Separately, deep and surprising connections to random matrix theory have led to generating functions for map enumeration. These generating functions are series, one for each genus $g$, whose $j^{th}$ Taylor coefficient counts the number 
of labelled maps on the surface with $j$ vertices of prescribed valence. One of the earliest approaches was based on a formal application of resolvent identities for random matrices that goes back to Ambjorn, Chekov, Kristjansen, and Makeenko \cite{bib:a}.  This is known as the method of {\it loop equations}. Eynard  \cite{bib:e,bib:ebook} subsequently improved on this work to establish a direct connection between loop equations and Tutte's equations that are key to the combinatorial method mentioned earlier.
Finally, in \cite{bib:biz} and, later in \cite{bib:fik}, a different random matrix approach to deriving generating functions was developed based on recurrence relations for orthogonal polynomials. Subsequently, a rigorous basis for deriving map generating functions in general was established in \cite{bib:em, bib:emp08, bib:bd, bib:ep, bib:ew}, and led to further insights into their structure. 
The present work builds on these and recent results of the authors to compare two expansions, both centered on recurrence coefficients for orthogonal polynomials. One of the expansions considers these coefficients in terms of their combinatorial interpretation related to graphical enumeration discussed above. The other understands these coefficients in terms of an orbit embedded in a dynamical system known as the discrete Painlev\'e I equation. Comparing these two expansions in a region where they are both valid, as illustrated in Figure \ref{fig:val}, provides a procedure to systematically count the number of regular $g$-maps with fixed number of vertices, for arbitrary values of $g$. This procedure builds on an approach first developed in \cite{bib:tip20} (Section 7.4).

The rest of this article is organized as follows. Section \ref{sec:two_exp} introduces the two expansions, which we call the \textit{genus expansion} and the \textit{center manifold expansion}. Section \ref{sec:bridge} recasts them using the same gauge as the asymptotic parameter $n \to \infty$, and identifies a common region of validity where they can be equated term by term. Section \ref{sec:closed_forms} uses the result of Section \ref{sec:bridge} to find closed-form expressions for the generating functions of labeled $g$-maps with 4-valent vertices, and illustrates the methodology in calculating the number of $g$-maps with up to 15 vertices, for genera $g$ between 0 and 7. Section \ref{sec:conclusions} summarizes our results and considers a range of extensions. These include a generalization to $2\nu$-valent 2-legged maps that makes use of asymptotic expansions available in the literature in lieu of the center manifold expansion, possible extensions of the method of \cite{bib:elt22} to higher-order Painlev\'e equations, a discussion of triangulations, and the existence of closed-form expressions for the number of 4-valent $g$-maps with an arbitrary number of vertices. For clarity, the body of the article only considers 4-valent maps. Proofs of all of the theorems are presented in the appendices, in the more general case of $2\nu$-valent maps.

\begin{figure}[h] 
\begin{center} 
\includegraphics[width=.95\linewidth]{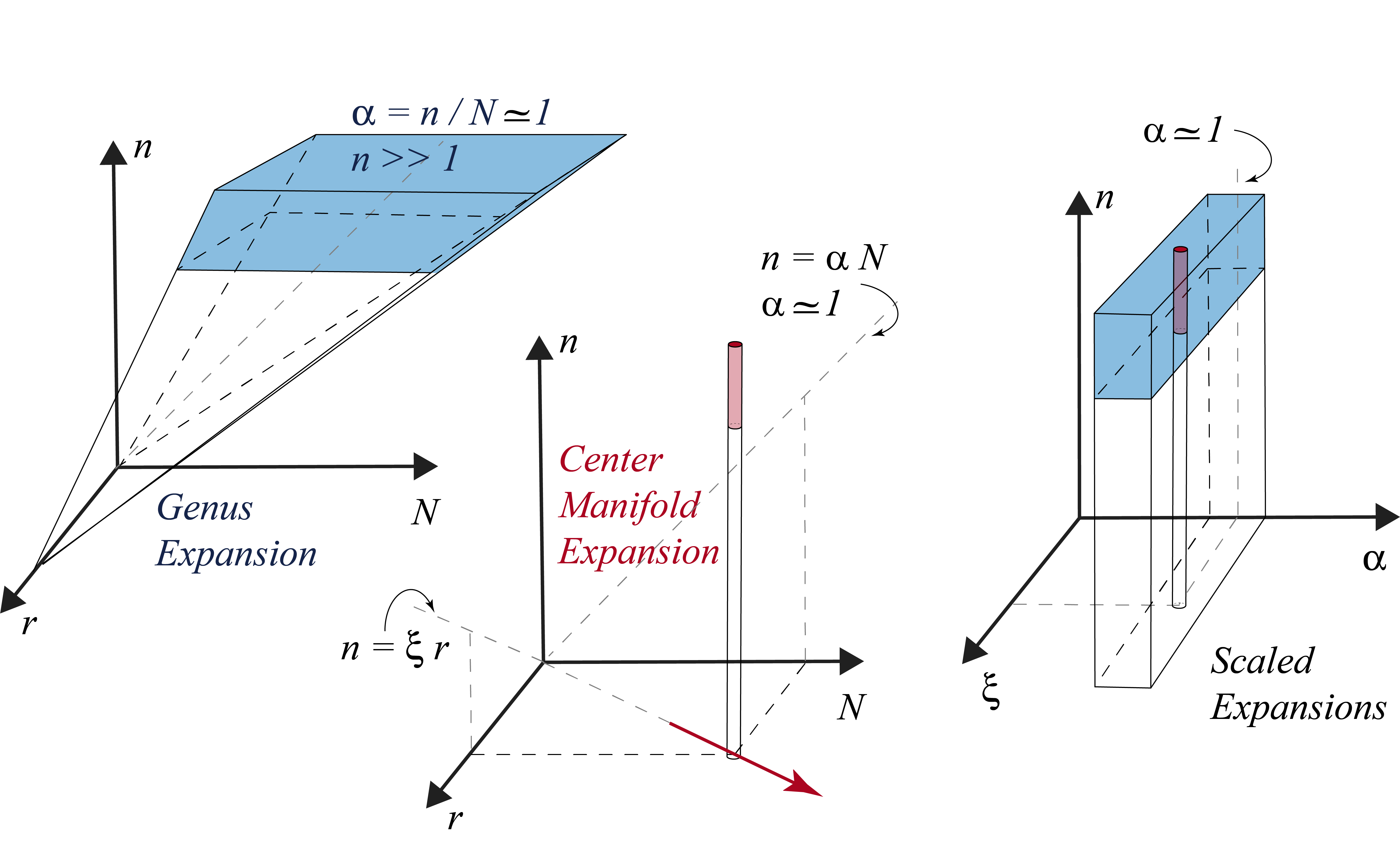}
\caption{The two expansions for $\nu=2$. Left: the genus expansion is valid as $n \to \infty$ for arbitrary values of $r >0$ and $\alpha = n / N \simeq 1$. Middle: the center manifold expansion is valid as $n \to \infty$ for arbitrary positive values of $r$ and $N$, here chosen such that $r = n / \xi$ and $N = n / \alpha$. As $n \to \infty$, both $r$ and $N$ increase linearly with $n$, as suggested by the red arrow. Right: in $(\xi,\alpha,n)$ coordinates, the regions of validity of the expansions overlap for fixed values of $\alpha \simeq 1$ and $\xi > 0$.}  \label{fig:val}
\end{center}
\end{figure}

\section{The Two Expansions}
\label{sec:two_exp}

\subsection{Recurrence Relations and The Genus Expansion}

We consider orthogonal polynomials defined on the real line with respect to an exponential weight of the form $ w(\lambda) = e^{- V_{{\bf t}, N}(\lambda)}$,
where the {\it potential} $V_{{\bf t}, N}$ is given by
\begin{equation} \label{potential}
V_{{\bf t}, N}(\lambda) = N \left(\frac1{2} \lambda^2 + \sum_{j=1}^J t_j \lambda^j\right), \qquad {\bf t}=(t_1, \cdots, t_J)
\end{equation}
with $J$ even. Although this paper will focus on a very particular case of (\ref{potential}), the general expression of $V_{{\bf t}, N}$ given above will be relevant in some of the appendices. Given the weight $w$, one can define a family of {\it monic} orthogonal polynomials $\{ \pi_\ell\}$ that satisfy the conditions
\begin{eqnarray*}
\int_{\mathbb{R}} \pi_n(\lambda) \pi_m(\lambda) w(\lambda) d \lambda = 0,\,\,\,\, n\ne m .
\end{eqnarray*}
When the potential $V_{{\bf t}, N}(\lambda)$ in \eqref{potential} is even, these polynomials are determined by a recurrence of the form
\begin{eqnarray} \label{rec}
\lambda\, \pi_n(\lambda) =  \pi_{n+1}(\lambda) + b^2_n\, \pi_{n-1}(\lambda).
\end{eqnarray}
The results directly pertinent to map enumeration rest on a detailed analysis of the truncated Mercer kernel associated to the family of monic orthogonal polynomials $\{ \pi_\ell\}$, 
\[
K_{{\bf t},n}(\lambda, \eta)= e^{ -(1/2) ( V_{{\bf t}, N}(\lambda) + V_{{\bf t},N}(\eta)   ) } \sum_{\ell = 0}^{n-1} \pi_\ell(\lambda) \pi_\ell(\eta), 
\]
and its large $n$ asymptotics. %
The fundamental result is the following so-called genus expansion.
\begin{theorem} \cite{bib:em} \label{Workhorse}
There exist $T >0$ and $\gamma > 0$ such that one has an asymptotic expansion, uniformly valid for $\alpha = \frac{n}N$ sufficiently close to 1 and all  ${\bf t} \in \mathbb{T}(T, \gamma) = \left\{  {\bf t} \in \mathbb{R}^J : \left|{ {\bf t}}\right| \leq T , \,\, t_J > \gamma \sum_{j=1}^{J-1} |t_j| \right\}$,  of the form
\beann
\int_{-\infty}^\infty F(\lambda) K_{{\bf t},n}(\lambda,\lambda)  d\lambda = F_0(\alpha, {\bf t}) + n^{-2} F_1(\alpha, {\bf t}) + n^{-4} F_2(\alpha, {\bf t}) + \cdots , 
\eeann
provided the function $F(\lambda)$ is $C^\infty$ and grows no faster than polynomially. The coefficients $F_m$ depend analytically on $\alpha$ and ${\bf t}$ for ${\bf t} \in \mathbb{T}(T, \gamma)$ and the asymptotic expansion may be differentiated 
term by term with respect to $\alpha$ and $\bf t$. 
\end{theorem}
 \noindent This is referred to as a {\it genus expansion} because for various choices of $F(\lambda)$ the coefficient of $n^{-2 g}$ is the generating function for some map enumeration problem on a surface of genus $g$.  
\begin{rem}
The discrete variable $n$ in this theorem, and the discussion preceding it, appears in other related contexts. In the setting of random matrix theory, briefly  mentioned in Section \ref{sec:mapsgraphs}, $n$ is the matrix size, and a probability density on 
$n \times n$ Hermitian matrices, $M$, is given by 
$\exp\left( -\mbox{Tr} \,\, V_{{\bf t}, N}(M) \right) dM$. In the dynamical setting of the discrete Painlev\'e I equation, to be discussed in Section \ref{sec:background}, $n$ labels the discrete time step. The parameters in {\bf t} of course determine the precise polynomial potential but, more importantly, they serve to identify different universality classes for statistical or dynamical behaviors of the physical system being modelled. Finally, the (continuous) parameter $N$ acts as a kind of inverse temperature in the random matrix setting and $\alpha = n/N$ is used to describe natural scaling invariances in all the systems just mentioned, as well as in this paper. In random matrix theory, $\alpha$ is called the {\it 'tHooft} parameter and is usually denoted by $x$. Here we use $\alpha$ to avoid confusion with the dynamic variable $x_n$ which will be introduced later.
\end{rem}

The particular form of the potential we will focus on for this paper is
\begin{eqnarray} \label{quarticwt}
V(\lambda)= N \left(\frac{1}{2}\lambda^2 + \frac{r}{4}\lambda^4 \right),
\end{eqnarray}
corresponding to ${\bf t}=(0,0,0,t) \in \mathbb{R}^4$, $t_4 = t = r/4$. Although focusing on this quartic case may seem restrictive from the viewpoint of general map enumeration, this was the case of original interest in the physics literature \cite{bib:biz}. For $V$ given by Equation \eqref{quarticwt}, we have the following result, obtained by setting $F(\lambda)= \lambda$ in Theorem \ref{Workhorse}, differentiating the resulting expansion term by term with respect to $t_1$ and then setting $t_1 = 0$.
\begin{theorem} \cite{bib:emp08}\label{thm:31}
For the recurrence coefficients $b_{n}^2$ of the three-term recurrence \eqref{rec}, associated to the weight with potential \eqref{quarticwt}, let $\alpha = n/N$ be in a neighborhood of 1, and let $t$ have positive real part. Then as $n\rightarrow \infty$, $b_{n}^2$ has an asymptotic expansion of the form
\begin{equation} 
b_{n}^2 = \alpha \left({z}_0(t,\alpha)+\frac{1}{n^2}{z}_1(t,\alpha) \cdots  \right), \label{eq:empex}
\end{equation}
uniformly valid on compact sets in $t$. The coefficients are analytic functions in a neighborhood of 0 with Taylor-Maclaurin expansion
$$
{z}_g\left(t,\alpha\right) = \sum_{j=0}^\infty (-1)^j \frac{\kappa^{(g)}_j}{j!} (\alpha t)^j
$$
 where 
$\kappa^{(g)}_j$ is the the number of labeled $g$-maps  with $j$ 4-valent vertices and exactly two vertices that are 1-valent. 
\end{theorem}

\noindent   A 1-valent vertex together with its unique edge is called a \textit{leg}. An example of a 4-valent, 2-legged, $g$-map is shown in Figure \ref{gmap}.

\begin{rem} \label{cartographic} %\label{rem:map_counts}
By this result, one may regard  ${z}_g(t, 1)$ as an \underline {exponential generating function} for counting inequivalent classes of 2-legged, 4-valent labelled  $g$-maps. Making our earlier variable replacement one has
\[
{z}_g\left(t,1\right) = {z}_g\left(\frac{r}{4},1\right) = \sum_{j=0}^\infty (-1)^j \frac{\kappa^{(g)}_j}{j!\, 4^j} r^j.
\]
Alternatively, one may consider  $\sum_{j=1}^\infty (-1)^j \hat{\kappa}^{(g)}_j r^j$, where
$ \hat{\kappa}^{(g)}_j =\frac{\kappa^{(g)}_j}{j!\, 4^j}$,
as an \underline{ordinary generating function} for {\rm unlabelled} 2-legged, 4-valent $g$-maps. Indeed, $j!$ is the size of the permutation group acting on vertex labels and $4^j$ is the size of the product of the cyclic groups acting on the distinguished edge labelling at each vertex. Then $j!\, 4^j$ is the cardinality of the orbit under the action of relabelling. This can be related to the action of the \underline{cartographic group} which acts as a subgroup of the group of permutations of all the half-edges, called \underline{darts}, attached to vertices. We refer the reader to \cite{bib:lz}, \cite{bib:emp08}(Section 5.10), and \cite{bib:vp} for more details on these matters, but the important upshot of these considerations is that due to the presence of legs in the maps being enumerated, there are no non-trivial equivalences of the type mentioned in section \ref{sec:mapsgraphs}. Consequently, $ \hat{\kappa}^{(g)}_j$  will always be an integer.
 In what follows we will be using $z_0(\frac{r}{4}, \alpha)$ where $z_0$ is uniquely determined by \eqref{eq:empex}. We note, however, that the coefficients in the Taylor-Maclaurin expansion of  $z_g(\frac{r}{4}, \alpha)$  alternate in sign and so must be respectively multiplied by $(-1)^j$ to recover 
$ \hat{\kappa}^{(g)}_j$.
\end{rem}

\begin{figure}[h] 
\begin{center} 
\includegraphics[width=.6\linewidth]{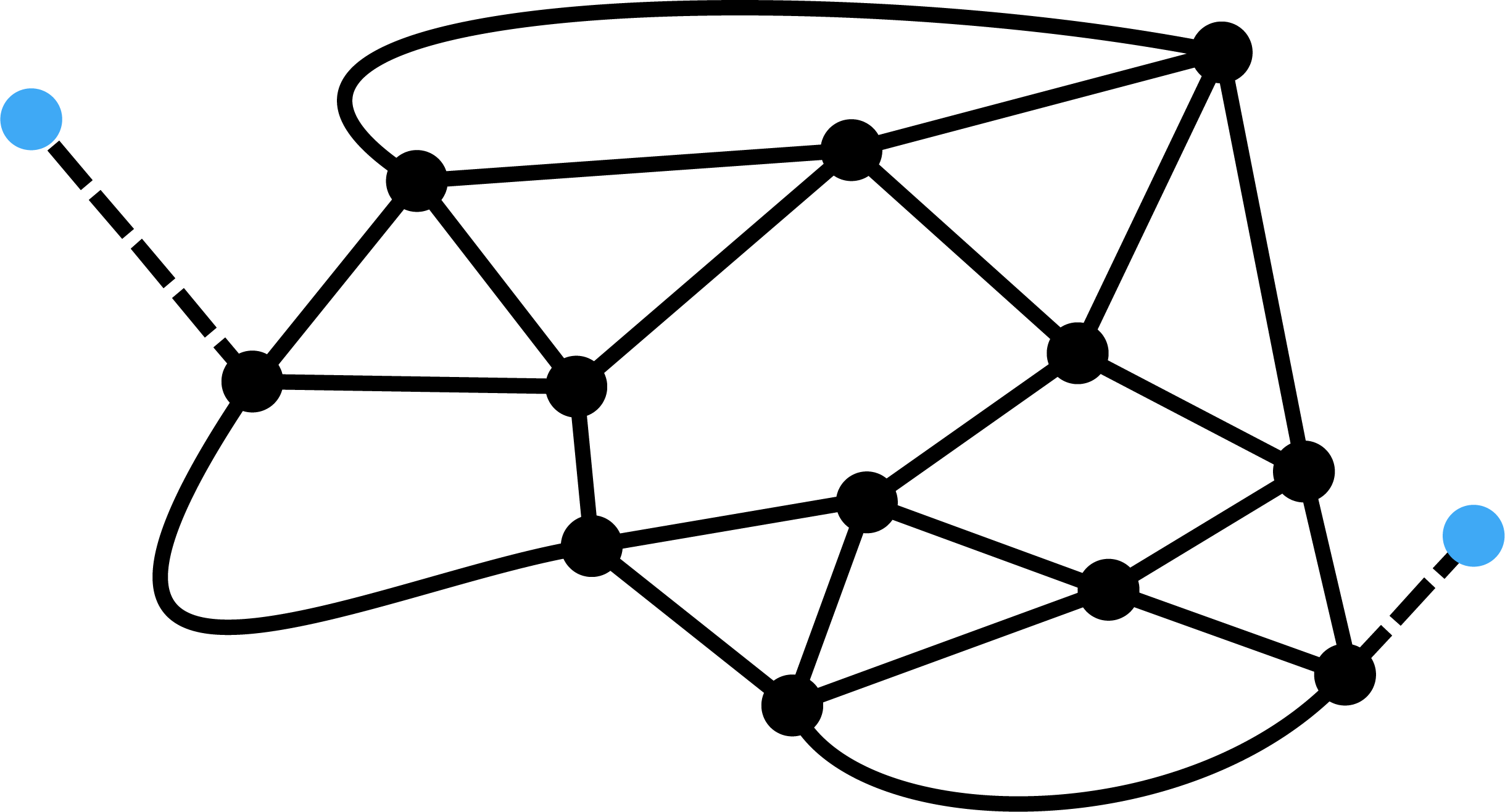}
\caption{Illustration of a 2-legged 4-valent map on the plane ($0$-map).}  \label{gmap}
\end{center}
\end{figure}
 
\noindent We will also make use of the following results, corresponding to Theorem B3 of \cite{bib:er}.
\begin{proposition}\label{cor:strip} \cite{bib:er} The asymptotic expansion (\ref{eq:empex}) is uniformly valid in a strip of constant width around the positive real $t$-axis. In addition, the coefficients $z_g(\frac{r}{4}, \alpha)$ have a maximal analytic continuation to the full complex $r$ plane minus the ray 
$(-\infty, -\frac{1}{12\alpha}]$.
\end{proposition}
\noindent We remark that this stated uniformity also follows independently from a  result due to Bleher and Its \cite{bib:bi}.
\medskip

Going further to solve for these generating functions, one can prove that the $z_g$ may be expressed as rational functions of $z_0$ \cite{bib:er}. In the case of $2\nu$-valent maps, this rational function takes the form:
\begin{equation} 
z_g=\dfrac{z_0(z_0-1)P_{3g-2}(z_0)}{(\nu-(\nu-1)z_0)^{5g-1}},
\label{eq:zgrationalform}
\end{equation} where $P_{3g-2}$ is a polynomial of degree $3g-2$, whose coefficients depend on $\nu$. This information is crucial for establishing our main results in Section \ref{sec:bridge} because it reduces the problem of finding $z_g$ to finding the \textit{finite} set of coefficients of $P_{3g-2}$. 
\begin{rem}
\label{rem:powers}
When $\nu = 2$, Equation \eqref{eq:zgrationalform} may be simplified as
\[
\label{eq:qpolyfac}
z_g=\dfrac{z_0(z_0-1)^{2 g} Q_{g-1}(z_0)}{(2-z_0)^{5g-1}},
\]
where $Q_{g-1}$ is a polynomial of degree $g-1$. An explanation is provided in Appendix \ref{app:z_counts}. Since this does not extend to the case $\nu > 2$, we continue our discussion of 4-valent maps by setting $\nu=2$ in the general form for $z_g$ stated in Equation \eqref{eq:zgrationalform}.
\end{rem}
The next statement concerns the structure of $z_0$. In \cite{bib:emp08}, the authors establish the form of $z_0$, which in our 4-valent case is given by: 
\begin{eqnarray} \label{z0}
z_0\left(\frac{r}{4},\alpha\right)=\dfrac{-1+\sqrt{1+ 12 \alpha r}}{6 \alpha r}, \qquad z_0\left(0,\alpha\right)= 1.
\end{eqnarray}
We note that $z_g$ is singular only at $z_0=2$, which corresponds to $r=-\frac1{12\alpha}$, consistent with the last statement of proposition \ref{cor:strip}.

In summary, the \textit{genus expansion} is the asymptotic expansion for the coefficients $b_n^2$ given in \eqref{eq:empex}, where the $z_g$ satisfy \eqref{eq:zgrationalform} and $z_0$ is expressed in \eqref{z0}. In particular, this expansion is uniformly valid for $\alpha = n/N$ sufficiently close to 1 and for all $r > 0$.

\subsection{The Center Manifold Expansion} \label{sec:background}
In \cite{bib:elt22}, we provided a dynamical systems description of certain non-polar orbits of the discrete Painlev\'e I Equation (dpI),
\begin{equation}
\label{eq:dpI}
x_{n+1}+x_n+x_{n-1} = \dfrac{n}{N\,r\,x_n}-\dfrac{1}{r}, \quad n \in \mathbb{N},\ x_n \in \mathbb{R}.
\end{equation}
Our focus was on solutions that remain positive for all $n \in \mathbb{N}$. 

\begin{rem} \label{rem:steq}
In Equation \eqref{eq:empex}, $z_0(t,\alpha)$ solves what is known as the string equation $1 = z_0(t,\alpha)+ 12\, t\, \alpha\, z_0^2(t,\alpha)$ (see Appendix \ref{sec:adapcent}).
From \eqref{eq:empex} one sees that as $n \to \infty$, $x_n \to \alpha z_0$. Applying this in \eqref{eq:dpI} and keeping in mind that $\alpha = n/N$, one immediately deduces that
$3 \alpha r z^2_0 + z_0 - 1 =0$, which is equivalent to the string equation with $t= r/4$. This shows that the string equation is nothing but the leading order form of the discrete Painlev\'e I equation in the continuum limit.
\end{rem}

It is natural to consider \eqref{eq:dpI} as a first order non-autonomous system in the $(x,y)$-phase plane given in terms of $(x_n,y_n) = (x_n, x_{n-1})$. In this formulation, the positivity condition of a solution becomes the requirement that it remains in the first quadrant. Such solutions are of particular interest since $x_n = b_n^2 > 0$ satisfies dpI when $b_n$ solves the recurrence relation 
\begin{equation}
\label{eq:rec2}
\lambda\, p_n(\lambda) = b_{n+1}\, p_{n+1}(\lambda) + b_n\, p_{n-1}(\lambda).
\end{equation}
In \eqref{eq:rec2}, the $p_n$ are orthonormal polynomials 
associated with the potential $V(\lambda)$ given in \eqref{quarticwt} and satisfy
\[
\int_{\mathbb{R}} p_n(\lambda) p_m(\lambda) w(\lambda) d \lambda = \delta_{nm}, \quad w(\lambda) = \exp\left(-N \left( \dfrac12 \lambda^2 + \dfrac{r}4 \lambda^4\right)\right).
\]
This should be contrasted with the monic orthogonal polynomials $\pi_n$ used in the previous section. However, the coefficients $b_n$ in \eqref{eq:rec2} are the same as in \eqref{rec}. We call the sequence of $x_n=b_n^2$ the Freud orbit, for $n>0$ \cite{bib:fre76}. The \textit{center manifold expansion} describes how $x_n = b_n^2$ depends on $n$ as $n \to \infty$ as a solution of dpI. It therefore provides information on the behavior of $b_n$ as $n \to \infty$ independently from the genus expansion. Matching the two in a region where they are both valid will give an expression for the coefficients of the polynomial $P_{3g -2}(z_0)$ appearing in Equation \eqref{eq:zgrationalform}, which in turn will lead to an expression for the generating functions $z_g$.

The approach of \cite{bib:elt22} in which the center manifold expansion is obtained, consists of the following elements. First, a change of variables 
\[
s = \frac{y}{x} + 1 + \frac{1}{r x}, \qquad f = \frac{n/N}{r x^2} - \frac{y}{x}, \qquad u = - \frac{1}{r x},
\]
transforms dpI, written as a 3-dimensional autonomous dynamical system in $(x,y,n)$ coordinates, into a system in $(s,f,u)$ coordinates that has two fixed points, $P_{-\infty}$ and $P_\infty$. Orbits that converge to $P_\infty$ (resp. $P_{-\infty}$) correspond to solutions of \eqref{eq:dpI} that grow without bounds as $n \to \infty$ (resp. $n \to -\infty$). Second, a proof that the Freud orbit converges to $P_\infty$, combined with compelling numerical evidence that this convergence occurs along the marginal eigendirection of the linearization about $P_\infty$, leads to the conjecture that the Freud orbit converges to $P_\infty$ along its center manifold. Third, an application of the center manifold theorem provides a Taylor expansion in powers of $u$ of the center manifold of $P_\infty$, valid to arbitrary order $p$:
\begin{equation}
\label{eq:cm}
s = s_\infty(u) = \sum_{j=1}^p s_j\, u^j + {\mathcal O}(u^{p+1}), \quad f = f_\infty(u) = \sum_{j=1}^p f_j\, u^j + {\mathcal O}(u^{p+1}).
\end{equation}
The coefficients $s_j$ and $f_j$ may be found explicitly order by order. Fourth, the change of variables from $(s,f,u)$ back to $(x,y,n)$, 
\[
x=-\frac{1}{r u}, \qquad y = - \frac{s+u-1}{r u}, \qquad \frac{n}{N} = \frac{s+f+u-1}{r u^2},
\]
requires that any orbit $\big(s_\infty(u_n),f_\infty(u_n),u_n\big)$ on the center manifold of $P_\infty$ should satisfy
\begin{equation}
\label{eq:unn_implicit}
\frac{n}{N} = \dfrac{s_\infty(u_n) + f_\infty(u_n) + u_n - 1}{r u_n^2} \\
\Longleftrightarrow \ \gamma\, n\, u_n^2 - u_n + 1 = s_\infty(u_n) + f_\infty(u_n),
\end{equation}
where $\gamma = r/N$ (for $\nu=2$).
 Finally, substituting a Laurent series in powers of $\sqrt n$ into the rightmost equation of \eqref{eq:unn_implicit} and solving term by term, leads to the following result.
\begin{theorem} \cite{bib:elt22}
In $(s,f,u)$ coordinates, the Freud orbit has the following asymptotic expansion
\begin{equation}
\label{eq:unn_explicit}
u_n = - \sqrt{\frac{3}{\gamma n}}-\frac{1}{2 \gamma  n} - \frac{1}{\left(8 \sqrt{3}\right) (\gamma n) ^{3/2}} + {\mathcal O}\left(n^{-5/2}\right) \quad \text{as } n \to \infty.
\end{equation}
\end{theorem}

\noindent This expansion may be continued to arbitrary order by appropriately selecting the order $p$ to which the Taylor expansions $s_\infty(u)$ and $f_\infty(u)$ are pushed in \eqref{eq:cm}. Moreover, because the Taylor remainder theorem provides control on the ${\mathcal O}(u^{p+1})$ terms in \eqref{eq:cm} as $u \to 0$, and because the Laurent series for $u_n$ on the Freud orbit is such that $u_n \to 0$ as $n \to \infty$, the expansion \eqref{eq:unn_explicit} is asymptotic as $n \to \infty$. Since $x_n = -1/(r u_n)$, Equation \eqref{eq:unn_explicit} leads to the \textit{center manifold expansion} of $x_n = b_n^2$ in powers of $n^{1/2}$. Although the existence of such an expansion was known \cite{bib:mnz85}, the dynamical systems context illuminates the special nature of the Freud orbit as a solution of dpI.

\section{Bridging the Two Expansions}\label{sec:bridge}

The previous section introduces two different asymptotic expansions of $x_n = b_n^2$ as $n \to \infty$, one arising from the setting of map enumeration, the other arising from dynamical systems theory. The aim of this section is to write these expansions in a common form, so that they can be equated. Caution should of course be exercised to ascertain that such a matching occurs in a region where both expansions are valid.

\subsection{Statement of both expansions}
 We start by recording both  expansions, to make clear which assumptions they involve and where they are valid. 

 \textit{The genus expansion.} Let $\alpha={n}/{N}$ be in a neighborhood of 1, and take $r>0$. Then Equation \eqref{eq:empex} tells us that the coefficients of the recurrence relation \eqref{rec}, which we now denote by $b_n^2 = x_{n,N,r}$, have the following expansion in terms of the generating functions $z_g(n,N,r)$:
\begin{equation} \label{eq:genusexpansion}
x_{n,N,r} \sim \alpha \left(\sum_{g=0}^\infty \dfrac{z_g(n,N,r)}{n^{2g}}\right),
\end{equation}
so that 
\begin{equation*}\label{eq:k1}
\left| x_{n,N,r}- \alpha \sum_{g=0}^m \dfrac{z_g(n,N,r)}{n^{2g}}\right| < \dfrac{K_{1,m}(N,r)}{n^{2m+2}}, \ \ \  m\geq0
\end{equation*}
uniformly for $\alpha \simeq 1$ and $r>0$. 
\begin{rem}
On the notation of constants: due to the presence of many error bounds, we will use the indexed constants $K_i$ or $K_{i,m}$ throughout the rest of this paper. The subscript $i$ will denote the order of appearance in this paper whereas $m$ will specify the largest index in the expansion.
\end{rem}

 \textit{The center manifold expansion.} For $N,r>0$ fixed, Equation \eqref{eq:unn_explicit}, together with the change of variable $x=-1/(r u)$, tells us that the recurrence coefficients $x_{n,N,r}$ are of the form
\begin{equation} \label{eq:centerxpansion}
x_{n,N,r}\sim \sum_{k=-1}^\infty \dfrac{c_k}{n^{k/2}},
\end{equation} so that 
\begin{equation} \label{eq:k2}
\left| x_{n,N,r} -\sum_{k=-1}^m \dfrac{c_k}{n^{k/2}} \right| < \dfrac{K_{2,m}(N,r)}{n^{(m+1)/2}}, \ \ \  m\geq-1.
\end{equation}
The nature of the dependence of $c_k$ and
$K_{2,m}$ on the parameters $N$ and $r$ will be revisited later in a rescaling argument.

There are two main challenges in relating these expansions. The first is that they are in different gauges, with the genus  expansion in $n^2$ versus the center manifold expansion in $\sqrt{n}$. The second challenge stems from their different regimes of uniform validity with respect to parameters; in the genus expansion, $N$ and $n$ go to infinity in a double scaling limit keeping $r > 0$ free, while for the center manifold expansion $N$ and $r$ are fixed but arbitrary as $n$ goes to infinity. We address both of these challenges in the next two sections, first by converting the genus expansion to the $\sqrt{n}$ gauge, and then by leveraging a rescaling argument for the center manifold expansion that allows us to send both $N$ and $r$ to infinity together with $n$, as depicted by the red arrow in the middle panel of Figure \ref{fig:val}.

\subsection{Genus Expansion in $\sqrt{n}$}
 Recall the explicit formula \eqref{z0} for the genus 0 generating function, which we now express as a function of $n$ and $\gamma = r/N$:
\begin{equation*}
z_0(n,\gamma)=\dfrac{-1+\sqrt{1+12 n\gamma}}{6  n\gamma}.
\end{equation*}
For $1/12<n \gamma$ it is straightforward to express $z_0$ as a convergent Laurent series in $\sqrt{n}$: 
\begin{equation}\label{eq:z0pu}
z_0(n,\gamma) = \sum_{i=1}^\infty a_{i,0} \left(\dfrac{1}{n\gamma}\right)^{i/2} = \sum_{i=1}^\infty a_{i,0}(\gamma) \left(\dfrac{1}{n}\right)^{i/2}.
\end{equation}
One can explicitly write the $a_{i,0}$ in terms of the Newton combinatorial coefficient $1/2 \choose i$, but for collecting terms in the expansion \eqref{eq:adapgenuspsnu2} below it will be more convenient to leave \eqref{eq:z0pu} as it is, with indices evident. Given the rational form \eqref{eq:zgrationalform} for $z_g$ in terms of $z_0$ (recall that for clarity $\nu$ is set equal to $2$ in the body of this article),
\begin{equation} \label{eq:rat2}
z_g=\frac{z_0(z_0-1)P_{3g-2}(z_0)}{(2-z_0)^{5g-1}},
\end{equation}
one derives a similar convergent series in $\sqrt{n}$ for the $z_g$:
\begin{equation}\label{eq:zgexphalf}
z_g(n,\gamma) = \sum_{i=1}^\infty a_{i,g} \left(\dfrac{1}{n\gamma}\right)^{i/2} =\sum_{i=1}^\infty a_{i,g}(\gamma) \left(\dfrac{1}{n}\right)^{i/2}.
\end{equation} 
The derivation is a simple application of the substitution of convergent series, whose validity for $1/12<n \gamma$ becomes apparent once one notes that $z_0$ is bounded between 0 and 1 for positive $n,N,r$. Let us denote the unknown coefficients of the polynomial $P_{3g-2}$ as follows:
\begin{equation*}
P_{3g-2}(z_0)=\beta_{0,g}+\beta_{1,g}z_0 + \cdots +\beta_{3g-2,g}(z_0)^{3g-2}.
\end{equation*}
\noindent 
\begin{rem}\label{rem:abetatri}
For $i \leq 3g-1$,  $a_{i,g}$ takes the form:
\begin{equation}
a_{i,g} =\left(\dfrac{-1}{2^{5g-1}3^{i/2}} \right)\beta_{i-1,g}+ L_{i,g}\label{eq:aigleadinglemma},
\end{equation} where $L_{i,g}$ is linear in $\beta_{j,g}$ for $j< i-1$. Proof of this fact is a consequence of the more general result for valence $2 \nu$ (see Lemma \ref{lemma:zgcoeffdependence} of Appendix \ref{sec:adaptgenus}) and simply amounts to collecting terms at the appropriate order. Given the form of dependence described in equation \eqref{eq:aigleadinglemma}, one finds that solving for these $\beta_{i,g}$ is achieved by solving a simple triangular system. In the appendices, the quantities $a_{i,g}(\gamma) := a_{i,g}/\gamma^{i/2}$ are denoted by $a_{i,g,2}(\gamma)$ to indicate that $\nu = 2$.
\end{rem}  

With the $z_g$ expressed as convergent series in inverse half powers of $n$, we can derive the following bivariate expansion for $x_n$ in the asymptotic gauge $\sqrt{n}$.

\begin{lemma}\label{lemma:bivariate}
Let $\alpha= \frac{n}{N}$ be in a neighborhood of 1, and let $\xi= \frac{n}{r}$ be fixed or bounded. Then, for $n$ large, $x_{n}$ has an asymptotic expansion as $n,N, r  \rightarrow \infty$, at relative rates given by $\alpha$ and $\xi$. The precise meaning of this is as follows. Define the partial sums 
\begin{equation}\label{eq:adapgenuspsnu2}
    \mathcal{G}^{(d)}:=\alpha\sum_{j=0}^{\lfloor \frac{d-1}{4} \rfloor}\sum_{i=1}^{d-4j}\dfrac{a_{i,j}(\gamma)}{n^{2j+i/2}}.
\end{equation}
Then for $n$ large we have the approximation:
\begin{equation}
|x_{n,\frac{n}{\alpha},\frac{n}{\xi}} - \mathcal{G}^{(d)}|<  \dfrac{K_{3,d}(\alpha, \gamma)}{n^{{(d+1)}/{2}}}, \ \ \  d\geq1. \label{eq:doublexpansionerror}
\end{equation}
In the above, the parameter $\gamma = r/ N = \alpha/\xi$ is finite and independent of $n$.
\end{lemma}
\noindent The generalization of Lemma \ref{lemma:bivariate} to valence $2 \nu$ is proven in Lemma \ref{lemma:genusadaptbound} of Appendix \ref{sec:adaptgenus}. 

\subsection{Scaling Properties of the Center Manifold Expansion}

We now move to extending the regime of the center manifold expansion to variable $N$ and $r$, in order to allow these parameters to tend to infinity, so that we may equate coefficients with those of \eqref{eq:adapgenuspsnu2}. First, observe the following rescaling of the Freud orbit:

\begin{lemma}\label{lemma:orbitrescale}
The Freud orbit satisfies the rescaling relation
\begin{equation}
    x_{n,N,r} = \frac{1}{n} x_{n,\frac{N}{n},\frac{r}{n}}.
\end{equation}
\end{lemma}
\noindent The proof of this lemma follows as a special case of Theorem \ref{thm:orbitrescale} found in Appendix \ref{sec:generalnuproof}. 
The explicit dependence of the coefficients $c_k$ of the center manifold expansion \eqref{eq:centerxpansion} in terms of $N$ and $r$ is given by the following lemma.

  \begin{lemma}\label{lemma:nrdegreesub}
The coefficients $c_k(N,r)$ of the center manifold expansion \eqref{eq:centerxpansion} satisfy the rescaling condition
\begin{equation}
    c_i( \sigma N, \sigma r) = \dfrac{1}{\sigma}c_i(N,r).
\end{equation}
\end{lemma}
\noindent The proof is given for the more general case of even, regular valence  in Lemma \ref{lemma:appnrdegreesub} of Appendix \ref{sec:adapcent}. Letting $n \to \infty$ while keeping $\alpha = \frac{n}{N}$ and $\xi = \frac{n}{r}$ constant in the above lemmas leads to the following theorem. 

\begin{theorem}\label{lemma:bmnnew}
Let $\alpha$ be in a neighborhood of 1, and let $\xi$ be bounded above and away from 0. Then, for $n$ large we have the following approximation:
\begin{equation}\label{eq:nu2centerrescaled}
   \left|x_{n,\frac{n}{\alpha},\frac{n}{\xi}}- \sum_{k=-1}^m \dfrac{c_{k}(\frac{1}{\alpha},\frac{1}{\xi})}{n^{1+k/2}} \right|
< \dfrac{K_{2,m}(1/\alpha,1/\xi)}{n^{(m+3)/2}}, \ \ \  m\geq-1.
\end{equation}
\end{theorem}
\noindent As before, we provide the proof, for general even valence, in Lemma \ref{lemma:centeradaptbound} of Appendix \ref{sec:compare}. This result enables a single vertical ray in the middle panel of Figure \ref{fig:val} to be extended along the direction given by the red arrow.

\subsection{Comparison of the two expansions}
\label{sec:comps}

Thus far, we have reformulated the center manifold and genus expansions to use the same gauge in an overlapping parameter regime. The following theorem establishes the equivalence of the adapted expansions.
\begin{theorem}\label{thm:comparison}
Let $\alpha=\dfrac{n}{N}$ be in a neighborhood of 1, and let $\xi=\dfrac{n}{r}$ be bounded above and away from 0. Then, for $n$ large the difference between the genus expansion in gauge $\sqrt{n}$ and the center manifold expansion can be bounded
 \begin{equation}
\left|\mathcal{G}^{(m+2)}- \sum_{k=-1}^{m} \dfrac{c_k(\frac{1}{\alpha},\frac{1}{\xi})}{n^{1+k/2}} \right|  <  \dfrac{K_{4,m}(1/\alpha,1/\xi)}{n^{(m+3)/2}} \label{eq:app72thm}, \ \ \  m\geq-1.
\end{equation}  
\end{theorem}
\noindent A simple proof using the triangle inequality to combine previously established estimates, \eqref{eq:doublexpansionerror} and \eqref{eq:nu2centerrescaled},   is provided in Theorem  \ref{thm:equivalentexpansions} of Appendix \ref{sec:compare}. This will establish the equivalence of these two asymptotic sequences.

\section{Closed-form expressions for $z_g$, $e_g$ and map counts}
\label{sec:closed_forms}

We are now ready to extract the coefficients of $P_{3g-2}$ from the two expansions of $x_{n,N,r}$, in order to obtain a closed-form expression of $z_g$ in terms of $z_0$ for each value of $g$. Generating functions for 4-regular (without legs) maps, $e_g$, are obtained from the $z_g$ by solving an inhomogeneous Cauchy-Euler equation, as described in \cite{bib:emp08} and \cite{bib:er14}.

\subsection{Closed-form expressions for $z_g(z_0)$}
We note that a \textit{finite} truncation of the center manifold expansion is sufficient to solve for the $z_g$. This follows from two essential facts about these expansions. First, the rescaled center manifold expansion \eqref{eq:centerxpansion} is already written in the $\sqrt n$ gauge and it is therefore immediate to identify which terms should be equated to those in the genus expansion  \eqref{eq:adapgenuspsnu2}. Second, and critically, for a fixed genus $g$, the polynomial $P_{3g-2}$ has a finite number ($3g-1$ to be precise) of  unknown coefficients $\beta_{i,g}$ and solving for these $\beta_{i,g}$ amounts to solving the simple triangular system discussed in remark \ref{rem:abetatri}. Tracking the order of the first occurrence of $\beta_{3g-2,g}$, we show in Appendix \ref{subsec:extracting} equation \eqref{eq:knucount} that we must include terms in the center manifold expansion up to $k=k_\nu$, where
\[
k_\nu = 5g\nu-2\nu-3g+1,
\]
 to obtain the expression of $z_g$ in terms of $z_0$. When $\nu = 2$, $k_2=7g-3$.

\begin{rem}
Using the factored form of $z_g$ provided in \eqref{eq:qpolyfac}, one can significantly reduce the number of terms needed to solve for $z_g$, down to $k_2=5g-2$.
\end{rem}
We briefly illustrate the matching process for $z_1$. The first few terms of the center manifold expansion  \eqref{eq:centerxpansion} are 
\[x_n \sim \frac{\sqrt{n}}{\sqrt{3rN}}-\frac{1}{6 r}+\frac{\sqrt{N}}{24\sqrt{3nr^3}}+\frac{\left(\frac{1}{n}\right)^{3/2} \left(48
   r^2-N^2\right)}{1152 \sqrt{3} r^2 \sqrt{r N}}-\frac{1}{144 r n^2}+ \cdots.
\]
For $\alpha$ and $\xi$ fixed as $n \to \infty$, this leads to
\begin{equation}
   x_n = \alpha\left(\dfrac{\gamma^{-1/2}}{3n^{1/2}}-\dfrac{\gamma^{-1}}{6n}+\dfrac{\gamma^{-3/2}}{24\sqrt{3}n^{3/2}}+\dfrac{48\gamma^{-1/2}-\gamma^{-5/2}}{1152\sqrt{3}n^{5/2}}-\dfrac{\gamma^{-1}}{144n^3}\right) + {\mathcal O}(n^{-7/2}).
\label{eq:s4} 
\end{equation}
On the other hand, the genus expansion reads
\begin{align}
\label{eq:gammanaddress}
    x_n = \alpha&\left(\dfrac{a_{1,0}}{(\gamma n)^{1/2}} +\dfrac{a_{2,0}}{(\gamma n)^{2/2}}+\dfrac{a_{3,0}}{(\gamma n)^{3/2}}+\dfrac{a_{4,0}}{(\gamma n)^{4/2}} +\dfrac{a_{5,0}}{(\gamma n)^{5/2}} \right.
    +\dfrac{a_{6,0}}{(\gamma n)^{6/2}} \nonumber \\ &+ \left. \dfrac{a_{1,1}}{\gamma^{1/2}n^{5/2}}+\dfrac{a_{2,1}}{\gamma^{2/2}n^{6/2}}+ {\mathcal O}(n^{-7/2})\right).
\end{align}
Our goal is to compare the two expansions above to find $P_{3\cdot 1-2}=P_{1}$, with two unknowns coefficients, $\beta_{0,1}$ and $\beta_{1,1}$, which can be obtained from $a_{1,1}$ and $a_{2,1}$.  From the terms of degrees $\gamma^{-1/2}{n^{-5/2}}$ and  ${\gamma^{-1}}{n^{-3}}$, we readily see that
 \begin{equation*}
     a_{1,1}=\dfrac{48}{1152\sqrt{3}},  \ \  \ \ a_{2,1}=\dfrac{-1}{144}.
 \end{equation*} Relating $a_{1,1}$ and $a_{2,1}$ back to $\beta_{0,1}$ and $\beta_{1,1}$ (as described in Remark \ref{rem:abetatri}), we have the triangular system:
 
\begin{equation*}
\dfrac{48}{1152\sqrt{3}} =\dfrac{-\beta_{0,1}}{16\sqrt{3}}, \ \ \ \
\dfrac{-1}{144}=\dfrac{-(\beta_{0,1}+2\beta_{1,1})}{96}
\end{equation*}
whose unique solution is  $\beta_{0,1}=-2/3$ and $\beta_{1,1}=2/3$. Thus, we obtain
\begin{equation}
z_1=\dfrac{z_0(z_0-1)(-\frac{2}{3}+\frac{2}{3}z_0)}{(2-z_0)^4} = \dfrac{2 z_0(z_0-1)^2}{3 (2-z_0)^4} \label{eq:z1beta},
\end{equation} which agrees with the computation of $z_1$ found in \cite{bib:emp08}. 

The methodology introduced in this article works because we can equate coefficients in the genus and center manifold expansions, once these are truncated to a particular order. The coefficients $a_{i,g}$ are thus obtained by equating two bivariate polynomials in $1/\sqrt \gamma$ and $1/\sqrt n$. From a computational point of view, it is easier to set $\alpha = 1$ and rewrite these polynomials as functions of $1/\xi = \gamma$ and $1/\sqrt r = \sqrt{\xi/n}$. As explained at the end of Section \ref{subsec:extracting}, this transformation is such that the unknowns $a_{i,g}$ only appear in terms that involve $\xi^{-2g}$ (when $\nu=2$), thereby making it easier to locate those coefficients in the truncated expansions. For illustration, in the above example this change of variable gives $\gamma^{-1/2}{n^{-5/2}} = \xi^{-2} r^{-5/2}$ and  ${\gamma^{-1}}{n^{-3}}=\xi^{-2} r^{-3}$, where the exponent of $1/\xi$ is equal to $2 g$ in both terms. The closed-form expression of any $z_g$ in terms of $z_0$ may be obtained by equating the relevant terms in the two expansions and solving for the $\beta_{i,g}$. Below, we give expressions for $z_2$ through $z_7$, which were derived in this manner, with the help of Mathematica \cite{bib:mathematica}.

 \begin{equation}
 z_2 
 %=\dfrac{z_0(z_0-1)}{(2-z_0)^9}\left(\dfrac{56}{9}-\dfrac{98}{3}z_0+\dfrac{182}{3}z_0^2-\dfrac{434}{9}z_0^3+14z_0^4\right) 
 = \frac{14 z_{0} \left(z_{0}-1\right)^{4} \left(9 z_{0}-4\right)}{9 \left(2-z_{0}\right)^{9}}. \label{eq:z2}
 \end{equation} 

\begin{align}
 z_3
 %=\dfrac{z_0(z_0-1)}{(2-z_0)^{14}} & \left( - \dfrac{592}{9}+\dfrac{35344}{27}z_0-\dfrac{182468}{27}z_0^2+\dfrac{444340}{27}z_0^3 -\dfrac{597400}{27}z_0^4+\dfrac{457976}{27}z_0^5\right.\\
 %& \left.\ \  -\dfrac{188404}{27}z_0^6+\dfrac{10796}{9}z_0^7\right). \nonumber
 = \frac{4 z_{0} \left(z_{0}-1\right)^{6} \left(8097 z_{0}^{2}-6616 z_{0}+444\right)}{27 \left(2-z_{0}\right)^{14}}.
\label{eq:z3}
\end{align}

\begin{align}
z_4 
%= \dfrac{z_{0} \left(z_{0}-1\right)}{\left(2-z_{0}\right)^{19}}
%& \left(-\frac{1487408}{81}+\frac{7045288}{81} z_{0}+\frac{4709776}{27} z_{0}^{2}-\frac{189927782}{81} z_{0}^{3}+\frac{654424778}{81} z_{0}^{4}\right.\nonumber\\
%& \ \ -\frac{414062362}{27} z_{0}^{5}+\frac{1477632562}{81} z_{0}^{6}-\frac{1134240178}{81} z_{0}^{7}+\frac{60954398}{9} z_{0}^{8}\\
%& \ \ \left.-\frac{50891926}{27} z_{0}^{9}+\frac{18696694}{81} z_{0}^{10}\right).\nonumber
= \frac{2 z_{0} \left(z_{0}-1\right)^{8}}{81 \left(2-z_{0}\right)^{19}} \left(9348347 z_{0}^{3}-10899460 z_{0}^{2}+1683284 z_{0}+743704\right).
\label{eq:z4}
\end{align}

\begin{align}
z_5 
%= \frac{z_{0} \left(z_{0}-1\right)}{\left(2-z_{0}\right)^{24}} 
%& \left(\frac{51209984}{9}-\frac{1853318432}{27} z_{0}+\frac{8932585312}{27} z_{0}^{2}-\frac{57414016576}{81} z_{0}^{3}\right.\nonumber\\
%& \ \ -\frac{700979860}{9} z_{0}^{4}+\frac{41977001332}{9} z_{0}^{5}-\frac{378735496496}{27} z_{0}^{6}+\frac{214230422672}{9} z_{0}^{7}\\
%& \ \ -\frac{80426510024}{3} z_{0}^{8}+\frac{563141631080}{27} z_{0}^{9}-\frac{301276928848}{27} z_{0}^{10}+\frac{106334958128}{27} z_{0}^{11}\nonumber \\
%& \ \ \left.-\frac{67107488084}{81} z_{0}^{12}+\frac{709788436}{9} z_{0}^{13}\right).\nonumber
= \frac{28 z_{0} \left(z_{0}-1\right)^{10}}{81 \left(2-z_{0}\right)^{24}}\left(228146283 z_{0}^{4}-343379456 z_{0}^{3}+89349936 z_{0}^{2}+50426664 z_{0}-16460352\right).
\label{eq:z5}
\end{align}

\begin{align}
z_{6} = 
\frac{4 z_{0} (z_{0}-1)^{12}}{729 (2-z_{0})^{29}}
&\left(7669263871659 z_{0}^{5}
-14108672477756 z_{0}^{4}+5354520803304 z_{0}^{3} \right. \nonumber\\
&\left. +2989338317984 z_{0}^{2}-2040880028176 z_{0}+236635393760\right).
\label{eq:z6}
\end{align}

\begin{align}
z_7 = \frac{8 z_{0} \left(z_{0}-1\right)^{14}}{2187 \left(2-z_{0}\right)^{34}}
& \big(8837111271832321 z_{0}^{6}-19191494504274856 z_{0}^{5}+9758098469191604 z_{0}^{4}
\nonumber \\
& +4849961265803344 z_{0}^{3}-5422884537586736 z_{0}^{2}+1237758341566528 z_{0} \label{eq:z7}\\
&-26678563494080\big).\nonumber
\end{align}
Consistent with Remark \ref{rem:powers}, each expression for $z_g$ above involves a polynomial, $Q_{g-1}$, of degree $g-1$ in $z_0$. These polynomials, normalized so that their $L^2$ norm over the interval $[-1, 1]$ is equal to $1$, are plotted in the top panel of Figure \ref{fig:poly}. We note that their roots are real and interlaced. The bottom panel shows normalized histograms of the roots with 5, 6, and 7 bins over the $[-1,1]$ interval, together with a possible limit of the empirical distribution of the zeros of the $Q_{g-1}$, given by
\[
q(x) = \frac{\exp(x - 1)}{\hbox{erf}(\sqrt{2}) \sqrt{\pi (1 - x)}}.
\]
Providing an explanation for these remarkable observations will be the subject of future exploration.

\begin{figure}[h] 
\begin{center} 
\includegraphics[width=.95\linewidth]{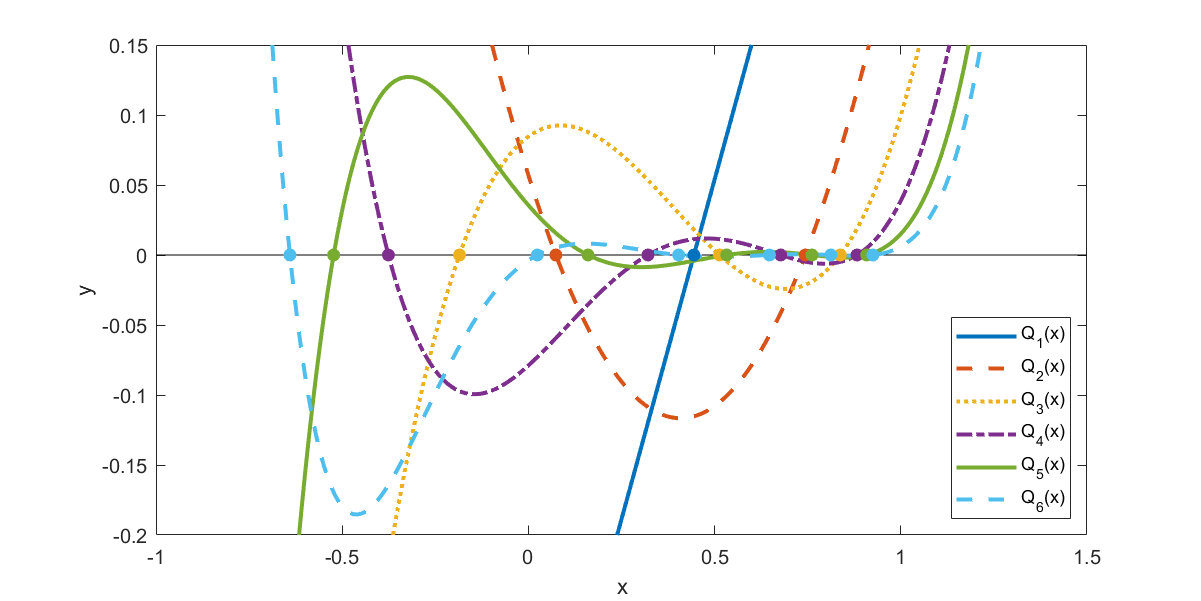}
\\
\includegraphics[width=.95\linewidth]{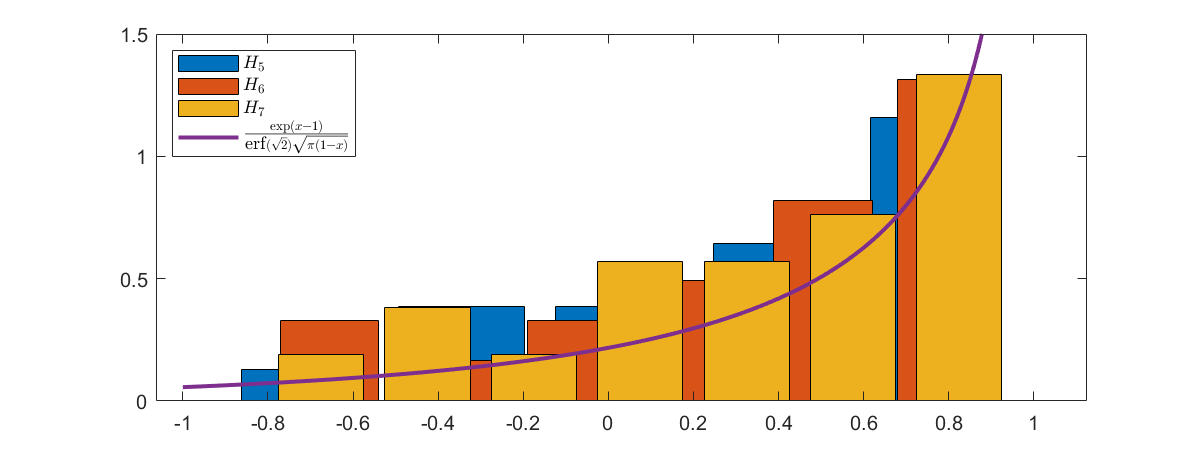}
\caption{Top: Graphs of the normalized polynomials $Q_{g-1}$, for $2 \le g \le 7$. Bottom: histograms of the zeros of these polynomials, using 5 ($H_5$), 6 ($H_6$), and 7 ($H_7$) bins, together with a possible asymptotic fit for the empirical distributions. }  \label{fig:poly}
\end{center}
\end{figure}

As indicated in Theorem \ref{thm:31}, explicit map counts for a surface of genus $g$ are obtained by repeated differentiation of $z_g$:
\begin{equation}
\label{eq:map_counts_z}
\text{Count of unlabeled 4-valent, 2-legged $g$-maps with $j$ vertices} =
\dfrac{(-1)^j}{j !} \left. \dfrac{d^j z_g}{d r^j}\right\vert_{r=0},
\end{equation}
where $z_g$ is expressed in terms of $z_0$, $z_0$ is a function of $r$ obtained from \eqref{z0} with $\alpha=1, \dfrac{d z_0}{d r} = 
-3 z_0^3/(2 - z_0)$, and $z_0 = 1$ when $r=0$. For reference, we used Maple \cite{bib:maple} to calculate the counts for genera 0 through 7 and a low number of vertices. These are recorded in Table \ref{table:mapcounttable} of Appendix \ref{app:z_counts}.

\subsection{Closed-form expressions for $e_g(z_0)$}

Using the procedure described in \cite{bib:emp08} and \cite{bib:er14}, the generating functions $e_g(z_0)$ for 4-valent $g$-maps can be recursively derived from the expressions for $z_g(z_0)$. Specifically, $e_g$ solves a forced Cauchy-Euler equation of the form 
\[
s^{2}\dfrac{d^{2} e_{g}}{d s^{2}} - 4 (g -1) s \dfrac{d e_{g}}{d s} + 2 e_{g} (2 g -1) (g -1)
 = \mathit{drivers}_{g},
\]
where $\mathit{drivers}_g$ is a function of $\big\{e_k(z_0)\big\}_{k < g}$ and $\big\{z_k(z_0)\big\}_{k \le g}$, $s=-t = -r/4$, $z_0$ is understood as a function of $s$, and $\dfrac{d z_0}{d s} = 12 z_0^3/(2 - z_0).$ Knowledge of the $z_k$ for $k \le g$ is therefore sufficient to obtain $e_g$. Using Equations \eqref{z0} and  \eqref{eq:z1beta} through \eqref{eq:z7}, we find, with the help of Maple \cite{bib:maple}, the following generating functions $e_g$ for genera 0 through 7. 
\begin{equation}
e_{0} = \frac{1}{2} \ln \! \left(z_{0}\right)+\frac{3}{8}-\frac{5 z_{0}}{12}+\frac{z_{0}^{2}}{24}. \label{eq:e_0}
\end{equation}
\begin{equation}
e_{1} = -\frac{1}{12} \ln \! \left(2-z_{0}\right). \label{eq:e_1}
\end{equation}
\begin{equation}
e_{2} =
%\frac{1}{240}-\frac{13}{72 \left(2-z_{0}\right)^{2}}+\frac{1}{2 \left(2-z_{0}\right)^{3}}-\frac{23}{48 \left(2-z_{0}\right)^{4}}+\frac{7}{45 \left(2-z_{0}\right)^{5}}.
= -\frac{\left(z_{0}-1\right)^{3} \left(3 z_{0}^{2}-21 z_{0}-82\right)}{720 \left(2-z_{0}\right)^{5}}.
\label{eq:e_2}
\end{equation}
\begin{align}
e_{3} = 
%& -\frac{1}{1008}+\frac{2741}{648 \left(2-z_{0}\right)^{4}}-\frac{731}{27 \left(2-z_{0}\right)^{5}}+\frac{3895}{54 \left(2-z_{0}\right)^{6}}-\frac{57980}{567 \left(2-z_{0}\right)^{7}}+\frac{35135}{432 \left(2-z_{0}\right)^{8}}\nonumber \\
%& -\frac{929}{27 \left(2-z_{0}\right)^{9}}+\frac{490}{81 \left(2-z_{0}\right)^{10}}.
-\frac{\left(z_{0}-1\right)^{5}}{9072 \left(2-z_{0}\right)^{10}}\left(9 z_{0}^{5}-135 z_{0}^{4}+855 z_{0}^{3}-2925 z_{0}^{2}-32704 z_{0}+17260\right).
\label{eq:e_3}
\end{align}
\begin{align}
e_{4} = 
%&\frac{1}{1440}-\frac{724697}{1944 \left(2-z_{0}\right)^{6}}+\frac{22543}{6 \left(2-z_{0}\right)^{7}}-\frac{43334935}{2592 \left(2-z_{0}\right)^{8}}+\frac{43141}{\left(2-z_{0}\right)^{9}}-\frac{76879901}{1080 \left(2-z_{0}\right)^{10}}\nonumber \\
%& +\frac{12625361}{162 \left(2-z_{0}\right)^{11}}-\frac{73403603}{1296 \left(2-z_{0}\right)^{12}}+\frac{711680}{27 \left(2-z_{0}\right)^{13}}-\frac{577556}{81 \left(2-z_{0}\right)^{14}}+\frac{1038212}{1215 \left(2-z_{0}\right)^{15}}.
-\frac{\left(z_{0}-1\right)^{7}}{38880 \left(2-z_{0}\right)^{15}}\big(27 z_{0}^{8}&-621 z_{0}^{7}+6426 z_{0}^{6}-39312 z_{0}^{5}+156870 z_{0}^{4}-423738 z_{0}^{3} \nonumber\\
& -13719796 z_{0}^{2}+12438536 z_{0}-1421392\big).
\label{eq:e_4}
\end{align}
\begin{align}
e_{5} =
%& -\frac{1}{1056}+\frac{283047593}{3888 \left(2-z_{0}\right)^{8}}-\frac{27042301}{27 \left(2-z_{0}\right)^{9}}+\frac{1523561099}{243 \left(2-z_{0}\right)^{10}}-\frac{7019254900}{297 \left(2-z_{0}\right)^{11}}+\frac{2152304035}{36 \left(2-z_{0}\right)^{12}}\nonumber\\
%& -\frac{8665997102}{81 \left(2-z_{0}\right)^{13}}+\frac{11254263191}{81 \left(2-z_{0}\right)^{14}}-\frac{395985656}{3 \left(2-z_{0}\right)^{15}}+\frac{236061593875}{2592 \left(2-z_{0}\right)^{16}}-\frac{10818910505}{243 \left(2-z_{0}\right)^{17}}\label{eq:e_5}\\
%& +\frac{395339812}{27 \left(2-z_{0}\right)^{18}}-\frac{706976564}{243 \left(2-z_{0}\right)^{19}}+\frac{21399280}{81 \left(2-z_{0}\right)^{20}}.
-\frac{\left(z_{0}-1\right)^{9}}{85536 \left(2-z_{0}\right)^{20}}\big(81 z_{0}^{11}&-2511 z_{0}^{10}+36045 z_{0}^{9}-317115 z_{0}^{8}+1906335 z_{0}^{7}-8258841 z_{0}^{6}\nonumber \\
&+26471691 z_{0}^{5}-63319725 z_{0}^{4}-6114807776 z_{0}^{3}+7592114712 z_{0}^{2}\\
&-1573981616 z_{0}-383964880\big).
\label{eq:e_5}
\nonumber
\end{align}

\begin{align}
e_{6} = 
-\frac{\left(z_{0}-1\right)^{11}}{79606800 \left(2-z_{0}\right)^{25}}\big(&167913 z_{0}^{14}-6548607 z_{0}^{13}+120225708 z_{0}^{12}-1379237382 z_{0}^{11} \nonumber \\
&+11065634613 z_{0}^{10}-65823071391 z_{0}^{9}+300177734274 z_{0}^{8} \nonumber \\
&-1069792529256 z_{0}^{7}+3007388659374 z_{0}^{6}-6676436144466 z_{0}^{5}\\
&-1987745167400532 z_{0}^{4}+3113497571095248 z_{0}^{3}\nonumber\\
&-955888270184512 z_{0}^{2} 
-369974786833952 z_{0}+139728961867968\big). \nonumber
\label{eq:e_6}
\end{align}

\begin{align}
e_{7} = 
-\frac{\left(z_{0}-1\right)^{13}}{104976 \left(2-z_{0}\right)^{30}}\big(729 z_{0}^{17}&-34263 z_{0}^{16}+766179 z_{0}^{15}-10836585 z_{0}^{14}+108693900 z_{0}^{13} \nonumber\\
&-821542176 z_{0}^{12}+4852565172 z_{0}^{11}-22920200316 z_{0}^{10} \nonumber\\
&+87835205250 z_{0}^{9}-275401525230 z_{0}^{8}+708906419910 z_{0}^{7}\\
&-1496166685650 z_{0}^{6}-1413940192593664 z_{0}^{5}+2672305782348584 z_{0}^{4} \nonumber\\
&-1119321797794336 z_{0}^{3}-479347256993504 z_{0}^{2}+370359088049920 z_{0} \nonumber\\
&-46240156833920\big). \nonumber
\label{eq:e_7}
\end{align}
While expressions for $z_j$ and $e_j$ with $j\le 3$ are known, we believe the above formulations for $z_4$ through $z_7$ and $e_4$ through $e_7$ are new.
Counts of unlabeled 4-valent maps are obtained by taking derivatives of $e_g$ with respect to $s$ and setting $s=0$ (recall that $z_0(0)=1$):
\begin{equation}
\label{eq:map_counts_e}
\text{Count of unlabeled 4-valent $g$-maps with $j$ vertices} = \dfrac{1}{4^j\,j !} \left. \dfrac{d^j e_g}{d s^j}\right\vert_{s=0}.
\end{equation}
These counts (obtained with Maple \cite{bib:maple}) are given in Table \ref{table:mapcounttable_e} of Appendix \ref{app:e_counts} for maps with up to 15 vertices on surfaces of genera $0 \le g \le 7$. Unlike the situation described in Remark \ref{cartographic} for the $z_g$, there are no legs in the enumerations corresponding to the $e_g$ to break symmetry. So there will be non-trivial equivalences, which are reflected in the fact that the unlabelled counts given by \eqref{eq:map_counts_e} are often rational numbers.

\subsection{ Comparison with known results in the literature}
\label{sec:comparison} We checked that the expressions for $z_1$, $z_2$, and $z_3$ provided above are equal to those given on pages 62, 63, and 66-67 of \cite{bib:emp08} for $\nu=2$\footnote{\label{footnote:2551}~Our analysis revealed a small typo in the expression for $z_{2,\nu}$ given in Section 5.4 on page 63 of \cite{bib:emp08}, where the coefficient 25551 should instead be 2551. Equations \eqref{eq:z2} and \eqref{eq:conjz2} are in agreement with the corrected expression.}. Because $z_4$ through $z_7$ are new, no direct comparisons are available. However, \cite{bib:er} (see also section \ref{sec:doublescale}) provides a recurrence formulation of the coefficient $a_{3g-1}^{(g)}(\nu)$ of $(\nu-(\nu-1) z_0)^{-(5 g -1)}$ in the partial fraction expansion of $z_g / z_0$ (Proposition 4.3 on page 511 of \cite{bib:er}),
\begin{equation} \label{eq:recursion}
a_{3(g + 1) - 1}^{(g + 1)}(\nu) = \dfrac{\nu^3 (25 g^2 - 1)}{6} a_{3 g - 1}^{(g)}(\nu) + \dfrac{\nu}{2} \sum_{m = 1}^g a_{3 m - 1}^{(m)}(\nu)\,a_{3 (g - m + 1) - 1}^{(g - m + 1)}(\nu),\qquad
a_{2}^{(1)}(\nu) = \dfrac{\nu^2}{6},
\end{equation}
which we confirmed was satisfied by the corresponding terms in all of the $z_g$ expressions presented in this article (for which $\nu = 2$).
Similarly, the expressions for $e_0$, $e_1$, and $e_2$ stated above are identical to those provided in \cite{bib:emp08} (pages 70, 71, and 72 with the constants $K_j$ set to 0) and to those on page 489 of \cite{bib:er}. A recent preprint by Bleher, Gharakhloo, and McLaughlin contains a closed-form formula for the number of labeled 4-valent maps on surfaces of genus 3 (\cite{bib:bgm} Theorem 1.6), which we used to check our expression for $e_3(z_0)$. In addition, for $g \ge 2$, we confirmed that the coefficients of the highest order terms in the partial fraction expansions of $e_g$ and $z_g/z_0$ in powers of $(2-z_0)^{-1}$ are related according to Equation (2-14) of \cite{bib:er14}, and that the constant term $a_0^{(g)}$ in the partial fraction expansion of $e_g$ satisfies the recurrence relation stated in Equation (2-15) of \cite{bib:er14},
\[
a_0^{(g)} = 
-2 \left(2 g -3\right)! \left(\frac{1}{\left(2 g +2\right)!}-\frac{1}{12 \left(2 g \right)!}+\frac{1}{\left(2 g -1\right)!} \sum_{k = 2}^{g -1} \frac{\left(\prod_{j =0}^{2 g -2 k +1}\left(2-2 k -j \right)\right) a_0^{(k)}}{\left(2 g -2 k +2\right)!}\right), \qquad g \ge 2,
\]
where the summation is set to zero for $g-1<2$.

Finally, the counts given in Tables \ref{table:mapcounttable} and \ref{table:mapcounttable_e} of Appendices \ref{app:z_counts} and \ref{app:e_counts} were compared to (and agreed with) the numerical values obtained by running a combinatorial code developed by V. Pierce \cite{bib:vp}, for $0 \le g \le 2$ and maps with up to 4 vertices. The algorithm underlying this code is based on cartographic group methods, mentioned in Remark \ref{cartographic}. Similarly, the number of labelled 4-valent 3-maps with 5 vertices was reported in Section 5.11 of \cite{bib:emp08} and is in agreement with Table \ref{table:mapcounttable_e}. We also checked that the results of Table \ref{table:mapcounttable_e} of the present article (in Appendix \ref{app:e_counts}) are in agreement with the counts shown in Table 2 of \cite{bib:dy17}, for genera 1 through 5 (see footnote\footnote{~For genus zero, we found a typo in Table 2 of \cite{bib:dy17} for row $k=8$: 154928203970560 should read 154948203970560. All of the other entries for $g=0$ agree with the present work.} for $g=0$). This provides a verification of the expressions for $e_4$ and $e_5$, and indirectly $z_4$ and $z_5$, since the former are obtained from the latter. The algorithm underlying the code used in \cite{bib:dy17} is based on the resolvent for the Lax difference operator appearing in the Toda Lattice equations associated to orthogonal polynomials. (We note that a similar algorithm was developed at the continuum limit level in \cite{bib:w}.) The relation between this resolvent and the discrete string equations used in the present paper is explained in Sections 1.2.2 and 4.1 of \cite{bib:ew}.

\section{Conclusions}
\label{sec:conclusions}

In this paper we have made a detailed comparison of two asymptotic expansions for the Freud orbit, a particular solution of the discrete Painlev\'e I equation (dpI): the genus expansion, which is based on a Riemann-Hilbert analysis of orthogonal polynomial systems, and the center manifold expansion, which is based on a dynamical systems analysis of dpI. The difference between them stems from the fact that the two expansions are obtained under different scaling limit assumptions, which have potentially different a priori parameter domains of validity in the large $n$ limit. However, in rescaling these expansions, we find there is a large overlap between their respective regions of uniform validity. Connecting the two expansions in this common parameter regime is the key technical mechanism that leads to the main result of the paper concerning map enumeration.  That result is two-fold. First, it provides an effective elementary means for counting the number of genus $g$, 4-valent maps with an arbitrary number of vertices. For illustration, counts of maps with up to 15 vertices on surfaces of genera 0 through 7 are provided in appendices \ref{app:z_counts} and \ref{app:e_counts}. Second, it yields an optimal bound on the finite number of steps required for evaluating all these counts, independent of the number of vertices.

The work presented here goes beyond the question of enumerating 4-valent $g$-maps. Indeed, the methodology we have introduced extends to maps of even valence $2 \nu$, through the lemmas provided in Appendices \ref{sec:adaptgenus} through \ref{sec:compare}. In addition, knowledge of how map counts change with parameters may provide insights into connections between generating functions and hierarchies of continuous Painlev\'e equations. Finally, the dynamical systems approach of \cite{bib:elt22}, which led to one of the asymptotic expansions used here, raises interesting questions on the role played by special solutions of dynamical systems in other areas of physics and mathematics. We elaborate on these ideas below.

\subsection{Generating Functions for Higher $\nu$}

\label{sec:highernu}
The present methodology may be extended to other forms of the potential
$V_{{\bf t}, N}(\lambda)$ (see Equation \eqref{potential}), as long as the two asymptotic expansions of $b_n^2$ are available. For the genus expansion, we have the  work of Ercolani, McLaughlin and Pierce \cite{bib:em,bib:emp08,bib:er} on 
potentials of the form
\[
V_{{\bf t}, N}(\lambda)= \left(\dfrac{\lambda^2}{2}+t_{2\nu}\lambda^{2\nu}\right). 
\] They proved that the expansion
\begin{equation}
b_{n,N}^2 \sim\dfrac{n}{N}\left(z_{0,\nu}(s)+\frac{1}{n^2}z_{0,\nu}(s) \cdots  \right) 
\label{eq:generalexp}
\end{equation} and rational expression 
\begin{equation}
z_{g,\nu}=\dfrac{z_{0,\nu}(z_{0,\nu}-1)P_{3g-2,\nu}(z_{0,\nu})}{(\nu-(\nu-1)z_{0,\nu})^{5g-1}}
\end{equation}
hold for these general types of potentials, where the $z_{g,\nu}$ are now generating functions for 2$\nu$-valent 2-legged $g$-maps, as stated in Equation \eqref{eq:zgrationalform}. The coefficients of $P_{3g-2,\nu}$ are still unknown for general $g$, although expressions for the generating functions $z_{g,\nu}$\, with $0 \le g \le 3$, are provided in \cite{bib:em,bib:emp08,bib:er}. For the center manifold expansion, we have the results of \cite{bib:mnz85}, which rely on an ordinary difference equation approach stemming from Poincar\'e-Perron type methods. They suffice to establish asymptotic expansions for the $b_n^2$, as needed for the higher $\nu$ case
considered in Appendices \ref{sec:adaptgenus} through \ref{sec:compare}.

As an example, Freud's equation for $\nu=3$ is
\begin{equation}
x_n+6t_6x_n(x_{n-1}x_{n-2}+x_{n-2}^2+x_{n}^2+2x_{n}x_{n-1}
+x_{n+1}x_{n+1}+2x_{n}x_{n+1}+x_{n+1}^2+x_{n+1}x_{n+2})=\dfrac{n}{N}, \label{eq:highernurecur}
\end{equation}
where $x_n = b_n^2$.
From Theorem 1 of M\'{a}t\'{e}, Nevai and Zaslavsky \cite{bib:mnz85} (see also \cite{bib:bmn88}),  we know that $x_n$ will have an asymptotic expansion in powers of $n^{1/3}$. 
Using Mathematica \cite{bib:mathematica} to compute this expansion to order $n^{-20/3}$ and mirroring the procedure described in Section \ref{sec:bridge}, we recover the closed form expression for $z_{1,3}$ derived in \cite{bib:emp08}. In addition this method provides the following result regarding $z_{2,3}$, also in agreement with \cite{bib:emp08} (see footnote \ref{footnote:2551}).

\begin{proposition}\label{conj:z23}
The generating function for labeled 6-valent, 2-legged maps on a genus 2 surface can be expressed as
\begin{align}
z_{2,3}&=\dfrac{z_{0,3}(z_{0,3}-1)}{(3-2z_{0,3})^9}\left(\dfrac{2673}{5}-\dfrac{62451}{20}z_{0,3}+\dfrac{25407}{4}z_{0,3}^2-\dfrac{27386}{5}z_{0,3}^3+\dfrac{8567}{5}z_{0,3}^4\right) \nonumber \\
& = \frac{z_{0,3} (z_{0,3}-1)^{2}}{20 \left(3-2 z_{0,3}\right)^{9}}\big(34268 z_{0,3}^{3}-75276 z_{0,3}^{2}+51759 z_{0,3}-10692\big),
\label{eq:conjz2}
\end{align}
where $z_{0,3}$ is the generating function for labeled planar 6-valent, 2-legged maps.
\end{proposition}

\subsection{Links with Higher-order Continuous Painlev\'e Equations} \label{sec:doublescale}

One of the principal interests and applications for the paper \cite{bib:er} was to provide a foundation for resolving the relation between a double scaling limit of dpI and the continuous Painlev\'e I equations that physicists had conjectured in some of the earliest explorations on quantum gravity \cite{bib:bk90,bib:fik91,bib:fikn06}. In \cite{bib:er} it was shown that the rational function in (\ref{eq:zgrationalform}) has a global Laurent polynomial representation of the form
\begin{eqnarray} \label{rational}
z_g (z_0) &=& z_0 \left\{ \frac{a_0^{(g)}}{(2 - z_0)^{2g}} + \frac{a_1^{(g)}}{(2 - z_0)^{2g+1}}+ \cdots + \frac{a_{3g-1}^{(g)}}{(2 - z_0)^{5g-1}}\right\}. \,\,\,\,\,\,\,\,\,\,\,\,\,\,\,\,\,\,\,\,
\end{eqnarray}
It was further established that in the double scaling limit for $\zeta = 2^{19/5} 3^{6/5} N^{4/5} (r + \frac1{12})$, as $N \to \infty$ and simultaneously $r \to - \frac1{12}$ from above, the sequence, in $g$, of top coefficients $a_{3g-1}^{(g)}$ precisely equals the coefficients of the asymptotic expansion of the {\it tri-tronquee} solution to the continuous Painlev\'e I equation, $d^2y/d\xi^2 = 6 y^2 + \xi$, \, in the {\it non-polar} sector. This analysis is the source  of the recursion formula \eqref{eq:recursion} for these asymptotic expansion coefficients,  which we used in section \ref{sec:comparison} to confirm our counts. It is natural to wonder if there are asymptotic structures of interest related to the lower coefficients in (\ref{rational}). That question continues to motivate applications of the explicit calculations carried out in the present paper.

The work in \cite{bib:er} also derives a novel extension of all these results to the general class of potentials of the form $V_{2\nu}(\lambda)= N \left(\frac{1}{2}\lambda^2 + \frac{r}{2\nu}\lambda^{2\nu} \right)$, but in which the Painlev\'e I equation is replaced by the $\nu^{th}$ equation in the associated {\it continuous} Painlev\'e I hierarchy. This is the continuous analogue of the hierarchy of discrete string, or Freud, equations mentioned in Appendix \ref{sec:adapcent}. A detailed exploration of the connections between generating functions and higher-order continuous Painlev\'e equations remains to be performed.

\subsection{Dynamical Systems perspective.} 

The bridge between the two expansions we have described corresponds to the unification of two perspectives: a Plancherel-Rotach type analysis initiated by Freud \cite{bib:fre76} and further developed by Nevai and co-authors \cite{bib:mnz85}, and more recent advances in Riemann-Hilbert analysis, related to integrable systems theory, as seen in the work of \cite{bib:fik} and \cite{bib:emp08}. The relation with discrete dynamical systems goes back to Freud who used that perspective to describe the leading order asymptotics of recurrence  coefficients  for families of orthogonal polynomials with exponential weights \cite{bib:fre76}. Later, motivations coming from random matrix theory and quantum gravity revived interest in these questions and led to re-interpretations  of discrete Painlev\'e equations as discrete string equations \cite{bib:fik}. We saw in section \ref{sec:doublescale} deep, physically meaningful, connections between multiple scaling limits of solutions to  discrete and continuous Painlev\'e systems. Such connections arise elsewhere in the literature \cite{bib:hfc} and it will be of interest to compare such results to our own.

Lew and Quarles \cite{bib:lq83} broadened the dynamical perspective for dPI to include other non-polar orbits, different from Freud's. More specifically, they used contraction mapping techniques to prove the existence of a one-dimensional family of solutions that remain positive under the dPI evolution. The overlap analysis presented in this paper solves, from a dynamical systems perspective, a connection problem for the non-polar solution between the regime $r >> 1$ where the purely quartic part of the potential is dominant and that near $r = -1/12$ related to the double-scaling limit mentioned at the start  of section \ref{sec:doublescale}. This has relevance for  non-perturbative string theory \cite{bib:bk90}. 

The global dynamical systems framework of \cite{bib:elt22} suggests two directions of future exploration. First, extending the analysis of \cite{bib:elt22} to general $\nu$ seems natural but presents some challenges, not least of which is that the phase space dimension of the dynamical system increases with $\nu$. However, for odd valence there is one important case, that of 3-valent ($\nu = 3/2$) graphs or, dually, triangulations that is dynamically tractable. As was the case for quadrangulations, formulas for $e_g(z_0)$ are already known for 3-valent maps when $0 \le g \le 2$ \cite{bib:ep,bib:bd13}. Corresponding counts (calculated with Maple \cite{bib:maple}) for graphs with up to 16 vertices are given in Appendix \ref{app:e_counts3}. The methodology introduced in the present article lays out a path toward obtaining counts for higher values of $g$. Interestingly, the results of \cite{bib:mnz85} do not help here since there is no corresponding family of classical orthogonal polynomials. However, our dynamical systems approach does apply, thereby providing a means to get a full asymptotic expansion of center manifold type. For instance, formally seeking an expression corresponding to the center manifold expansion leads to the following formula for $z_{2,3/2}$:
\begin{align*}
z_{2,3/2} & = \frac{z_{0,3/2} }{(z_{0,3/2}^2-3)^9}\left(\frac{243}{16} - 297 z_{0,3/2}^2 + \frac{15513}{16} z_{0,3/2}^4 - \frac{9705}{8} z_{0,3/2}^6 + \frac{9045}{16} z_{0,3/2}^8 + \frac{93}{4} z_{0,3/2}^{10} \right. \\ & \left. \qquad \qquad \qquad \qquad - \frac{993}{16} z_{0,3/2}^{12} - \frac{9}{8} z_{0,3/2}^{14}
\right)\\
&= -\frac{3 z_{0,3/2} \left(z_{0,3/2}^2-1\right)^{4}}{16 \left(z_{0,3/2}^{2}-3\right)^{9}}\left(6 z_{0,3/2}^{6}+355 z_{0,3/2}^{4}+1260 z_{0,3/2}^{2}-81\right).
\end{align*}
Combining the genus and center manifold expansions to obtain triangulation counts for topological surfaces of higher genus is something we will explore in future work. Indeed, triangulations are, from many mathematical perspectives, the class of maps of broadest interest.

Second, from the viewpoint of dynamics, many critical features go beyond all orders from what can be seen in just asymptotics. This was already evident in the pioneering work of Lew and Quarles \cite{bib:lq83} and such a realization is manifest in the results, both theoretical and numerical, found in \cite{bib:elt22}. Consequently, many of the algebraic structures we have been working with in this paper, such as string equations and generating functions, which are based on asymptotic expansions of particular orbits, necessarily extend to a plethora of other orbits that differ from the particular orbit only beyond all orders. This opens many avenues for dynamical and numerical exploration that we plan to pursue.

\subsection{Closed-form expressions for the map counts}
The number of regular $g$-maps may be obtained by taking successive derivatives of $z_g(z_0)$ or $e_g(z_0)$ and evaluating the result at $r=0$ or $s=0$ (corresponding to $z_0 = 1$), as indicated in Equations \eqref{eq:map_counts_z} and \eqref{eq:map_counts_e}. Knowledge of $z_g$ and $e_g$ as functions of $z_0$ is therefore sufficient to obtain such counts. A remaining challenge is to formulate the result as a closed-form expression that is solely a function of the regular valence $2 \nu$ and the number of vertices $j$. For instance, the number of $2\nu$-valent $0$-maps with $j$ vertices is \cite{bib:emp08}
\begin{equation}
\label{eq:zero_map_couts}
{\mathcal N}_{0,\nu} (j) = \left(2 \nu {2\nu -1 \choose \nu -1}\right)^j
\dfrac{(\nu j - 1)!}{\big((\nu-1) j + 2 \big)!}.
\end{equation}
A few similar results are known for low values of $g$ \cite{bib:bd13,bib:bgm}. When $\nu=2$, the expressions for $z_g$ and $e_g$ presented in this article may be used to derive closed-form expressions for ${\mathcal N}_{g,\nu} (j)$ for all genera for which $z_g(z_0)$ is known. This work is beyond the scope of the present article and will be described separately \cite{bib:elt23}.

In summary, the present article illustrates how results from Riemann-Hilbert analysis and either Plancherel-Rotach asymptotics or center manifold theory may be combined to provide a solution to a longstanding combinatorial problem in  map enumeration. In addition, including a dynamical systems perspective opens the door to further explorations that have the potential to reveal deep connections between various branches of mathematics.

\appendix

\section{Genus Expansion in the $n^{1/\nu}$ gauge}
\label{sec:adaptgenus}
In this appendix, we reformulate the genus expansion for $x_n$ in the gauge $n^{1/ \nu}$, where $\nu > 1$ is an integer.
From \cite{bib:emp08} we have the following (rescaled) polynomial equation, known as the {\it string equation}, that implicitly defines $z_{0,\nu}$:
\begin{equation} \label{eq:zggendefnoscale}
    1= z_{0,\nu}+ {2\nu -1\choose \nu -1}\left( \dfrac{n}{N} \right)^{\nu -1}r_{2\nu} z_{0,\nu}^{\nu},
\end{equation}
where $r_{2\nu}$ here is related to the variables in \eqref{potential} as $t_{2\nu}=r_{2\nu}/2\nu.$ The parameter $r$ used previously in this article may be expressed in this notation as $r=r_4$. Defining $\gamma := r_{2 \nu}/N^{\nu-1}$, we rewrite the string equation: 
\begin{equation} \label{eq:zggendef}
    1= z_{0,\nu}+ {2\nu -1\choose \nu -1}\gamma n^{\nu-1} z_{0,\nu}^{\nu}.
\end{equation}
Using the Newton-Puiseux theorem \cite{bib:bk86}, we can derive the following convergent expansion for $z_{0,\nu}$:
\begin{equation} \label{eq:zggenconv}
z_{0,\nu} = \sum_{i=\nu-1}^\infty a_{i,0,\nu}(\gamma)n^{-i/ \nu}.
\end{equation}
Note that the value of the lower bound of summation
reflects the balance in equation \eqref{eq:zggendef}.
Using the expansion \eqref{eq:zggenconv} for $z_0$ and the rational form for $z_g$:
\begin{equation}\label{eq:zggenrationaldef}
z_{g,\nu}=\dfrac{z_{0,\nu}(z_{0,\nu}-1)P_{3g-2,\nu}(z_{0,\nu})}{(\nu-(\nu-1)z_{0,\nu})^{5g-1}}
\end{equation}
where $P_{3g-2,\nu}$ is a polynomial in $z_{0,\nu}$ of degree $3g -2$, we can derive a convergent expansion for $z_{g,\nu}$: 
\begin{equation}\label{eq:zggenseriesdef}
z_{g,\nu} = \sum_{i=\nu-1}^\infty a_{i,g,\nu}(\gamma) n^{-i/ \nu},
\end{equation}
where the dependence of the coefficients $a_{i,g,\nu}$ on $\gamma$ has been made explicit.
Denote the coefficients of  $P_{3g-2,\nu}$ as $\beta_{i,g,\nu}$, for $0\leq i\leq 3g-2$.

\begin{lemma}
\label{lemma:zgcoeffdependence}
For $1\leq i \leq 3g-1$,  the coefficient $a_{i(\nu-1),g,\nu}$ takes the form
\begin{equation}
a_{i(\nu-1),g,\nu} =\left(\dfrac{-a^{i}_{\nu-1,0,\nu}}{\nu^{5g-1}} \right)\beta_{i-1,g,\nu}+ L_{i,g,\nu}\label{eq:aigleadinglemmagen},
\end{equation} where $L_{i,g,\nu}$ is linear in $\beta_{j,g,\nu}$ for $j< i-1$. Although this dependence is implicit, the coefficients $a_{k,g,\nu}$ are functions of $\gamma$.
\end{lemma}

\begin{proof}
With $a_{i(\nu-1),g,\nu}$ being the coefficient of $n^{-i(\nu-1) /\nu}$ in the expansion \eqref{eq:zggenseriesdef} for $z_{g,\nu}$, we simply need to collect terms in \eqref{eq:zggenrationaldef} at this order. We will view the rational function \eqref{eq:zggenrationaldef} as the product of two terms
\begin{align}
    P_{3g-2,\nu}(z_{0,\nu}) = \sum_{k=0}^{3g-2}\beta_{k,g,\nu}\, z_{0,\nu}^k &=
    \sum_{k=0}^{i-2}\beta_{k,g,\nu}\,n^{-k(\nu-1)/\nu}\Big(\sum_{j=0}^\infty a_{j+\nu-1,0,\nu}\, n^{-j/\nu}\Big)^k \nonumber\\
    &\ \ +\beta_{i-1,g,\nu}\,n^{-(i-1)(\nu-1)/\nu}\Big(\sum_{j=0}^\infty a_{j+\nu-1,0,\nu}\, n^{-j/\nu}\Big)^{i-1} \nonumber\\
    &\ \ +\sum_{k=i}^{3g-2}\beta_{k,g,\nu}\,n^{-k(\nu-1)/\nu}\Big(\sum_{j=0}^\infty a_{j+\nu-1,0,\nu}\, n^{-j/\nu}\Big)^k \nonumber\\
    &=  \Tilde{L}_{i,g,\nu} + \dfrac{a^{i-1}_{\nu-1,0,\nu}}{n^{(i-1)(\nu-1) /\nu}}\beta_{i-1,g,\nu} + \mathcal{O}\left(n^{-((i-1)(\nu-1)+1))/\nu}\right)
    \end{align}
and
    \begin{align}
    & \dfrac{z_{0,\nu}(z_{0,\nu}-1)}{(\nu-(\nu-1)z_{0,\nu})^{5g-1}} = \dfrac{-z_{0,\nu}}{\nu^{5 g - 1}} \big(1-z_{0,\nu}\big)\left(1-\dfrac{\nu-1}{\nu} z_{0,\nu}\right)^{1-5 g}\nonumber\\
    = & \dfrac{-a_{\nu-1,0,\nu}\, n^{-(\nu-1)/\nu} \left(1+{\mathcal O}\Big( n^{-1/\nu}\Big)\right)}{\nu^{5 g - 1}} \left(1-{\mathcal O}\Big( n^{-(\nu-1)/\nu}\Big)\right) \left(1- {\mathcal O}\Big( n^{-(\nu-1)/\nu}\Big)\right)^{1-5g}\nonumber\\
    = &\dfrac{-a_{\nu-1,0,\nu}}{\nu^{5g-1}n^{(\nu-1) /\nu}}+ \mathcal{O}\left(n^{-1}\right), \label{eq:ltilde}
\end{align}
where $\Tilde{L}_{i,g,\nu}$ are the terms collected from $\beta_{j,g,\nu}z_{0,\nu}^j$ (for $j< i-1$) in $P_{3g-2,\nu}$, which is linear by inspection. Since none of the terms in \eqref{eq:ltilde} depend on any of the $\beta_{j,g,\nu}$, the result follows.
\end{proof}

\begin{rem}\label{rem:coeffsuffice}
Per Lemma \ref{lemma:zgcoeffdependence}, knowing the $3g-1$ coefficients $\{a_{i(\nu-1),g,\nu}\}_{i=1}^{3g-1}$ in the expansion of $z_{g,\nu}$ guarantees that one can always solve for all $\beta_{i,g,\nu}$, through a non-singular triangular system. In turn, one then has the entirety of the $z_{g,\nu}$ expansion \eqref{eq:zggenseriesdef}, by expanding its rational form  \eqref{eq:zggenrationaldef}. 
\end{rem}

\begin{lemma}
Let $n,N,r_{2\nu}\rightarrow\infty$ at related rates $\alpha = n/N$ and $\xi = n^{\nu-1}/r_{2 \nu}$, then $x_{n,\frac{n}{\alpha},\frac{n^{\nu-1}}{\xi}}$ has an asymptotic expansion in this multi-scale regime of the form
\begin{equation}
\label{eq:shiftedgenuserrorbound}
\left| x_{n,\frac{n}{\alpha},\frac{n^{\nu-1}}{\xi}} - \alpha \sum_{g=0}^m \dfrac{z_{g,\nu}(\gamma)}{n^{2g}}\right| < \dfrac{K_{5,m}(\alpha,\gamma)}{n^{2m+2+\frac{\nu-1}{\nu}}}, \ \ \  m\geq0,
\end{equation}
where $\gamma = r_{2 \nu}/N^{\nu-1} = \alpha^{\nu-1} / \xi$ is independent of $n$.
\end{lemma}
\begin{proof}
First, similar to Theorem \ref{thm:31}, when $n,N\rightarrow\infty$ at the related rate $\alpha = n/N$, the coefficient $x_{n,N,r_{2\nu}}=x_{n,n/\alpha,r_{2\nu}}$ has an asymptotic expansion (the genus expansion) of the form \cite{bib:emp08}
\begin{equation}
\label{eq:rfixedasympbound}
\left| x_{n,\frac{n}{\alpha},r_{2 \nu}} - \alpha \sum_{g=0}^m \dfrac{z_{g,\nu}(\alpha,r_{2 \nu})}{n^{2g}}\right| < \dfrac{K_{1,m}(\alpha)}{n^{2m+2}}, \ \ \  m\geq0,
\end{equation}
where it was also shown in \cite{bib:er} that the constants $K_{1,m}$ are uniform for $r_{2\nu}>0$. This fact is critical as it allows us to vary $r_{2\nu}$ and still maintain control over the error in \eqref{eq:rfixedasympbound}. Applying the reverse triangle inequality with this bound, we have the inequalities (recall that $r_{2 \nu} = \frac{n^{\nu-1}}{\xi}$ by definition of $\xi$):
\begin{align*}
    \left\vert x_{n,\frac{n}{\alpha},\frac{n^{\nu-1}}{\xi}} - \alpha \sum_{g=0}^m \dfrac{z_{g,\nu}(\gamma)}{n^{2g}}\right\vert
     -\left \vert \alpha\dfrac{z_{m+1,\nu}(\gamma)}{n^{2m+2}} \right|
     & \leq  \left \vert \left(x_{n,\frac{n}{\alpha},\frac{n^{\nu-1}}{\xi}} - \alpha \sum_{g=0}^m \dfrac{z_{g,\nu}(\gamma)}{n^{2g}}\right)-\alpha\dfrac{z_{m+1,\nu}(\gamma)}{n^{2m+2}}\right\vert \\
    &= \left \vert x_{n,\frac{n}{\alpha},\frac{n^{\nu-1}}{\xi}} - \alpha \sum_{g=0}^{m+1} \dfrac{z_{g,\nu}(\gamma)}{n^{2g}}\right\vert \\
    & < \dfrac{K_{1,m+1}(\alpha)}{n^{2m+4}}.
\end{align*}
From the leading order term observed in \eqref{eq:zggenseriesdef}, we know that $z_{g,\nu}(\gamma) = a_{\nu-1,g,\nu}(\gamma) \left(n^{\frac{-(\nu-1)}{\nu}}\right) + O(n^{-1})$ and that $\gamma$ is independent of $n$ when $n,N,r_{2\nu}\rightarrow\infty$ at related rates $\alpha = n/N$ and $\xi = n^{\nu-1}/r_{2 \nu}$. With $\nu \geq 2$ we also have that $2m+4>2m+2+\frac{\nu-1}{\nu}$. Combining these facts with the previous inequality completes the proof:
\begin{align*}
    \left\vert x_{n,\frac{n}{\alpha},\frac{n^{\nu-1}}{\xi}} - \alpha \sum_{g=0}^m \dfrac{z_{g,\nu}(\gamma)}{n^{2g}}\right\vert   & < \left \vert \alpha\dfrac{z_{m+1,\nu}(\gamma)}{n^{2m+2}} \right|+\dfrac{K_{1,m+1}(\alpha)}{n^{2m+4}} < \dfrac{K_{5,m}(\alpha,\gamma)}{n^{2m+2+\frac{\nu-1}{\nu}}}.
\end{align*}

\end{proof}

\begin{lemma} \label{lemma:floorineq}
For $d,\nu \in \mathbb{Z}, d\geq \nu-1, \text{ and } \nu \geq 2$, we have the inequality

\begin{equation*}
    2\lfloor \frac{d-\nu+1}{2\nu} \rfloor+2+\frac{\nu-1}{\nu} \geq \dfrac{d+1}{\nu}.
\end{equation*}
\end{lemma}
\begin{proof}
Write 
\begin{equation}
    0 \le d-\nu+1 = 2\nu m + r, \text{ with } 0 \le r \leq 2\nu-1, \quad  r,m \in \mathbb{Z}
\end{equation} so that $\lfloor \frac{d-\nu+1}{2\nu} \rfloor =  m \ge 0$. Then,
\begin{equation*}
    \dfrac{d+1}{2\nu} = m +\dfrac{r+\nu}{2\nu}.
\end{equation*}
Thus, 
\begin{equation*}
\lfloor \frac{d-\nu+1}{2\nu} \rfloor+1+\frac{\nu-1}{2\nu}  \geq \dfrac{d+1}{2\nu}
\Longleftrightarrow m+1 +\frac{\nu-1}{2\nu}  \geq m +\dfrac{r+\nu}{2\nu} \Longleftrightarrow 2\nu-1 \geq r,
\end{equation*}
and the result follows.
\end{proof}

\begin{lemma}\label{lemma:genusadaptbound}
Let $n,N,r_{2\nu}\rightarrow\infty$ at related rates $\alpha = n/N$ and $\xi = n^{\nu-1}/r_{2 \nu}$ and define the partial sums
\begin{equation}\label{eq:adapgenusps}
    \mathcal{G}^{(d)}:=\alpha\sum_{j=0}^{\lfloor \frac{d-\nu+1}{2\nu} \rfloor}\sum_{i=\nu-1}^{d-2j\nu}\dfrac{a_{i,j,\nu}(\gamma)}{n^{2j+i/\nu}}.
\end{equation}
Then, these partial sums serve as an equivalent asymptotic sequence for $x_{n,N,r_{2\nu}}$ in this scaling limit, meaning that
\begin{equation}
    \left|x_{n,\frac{n}{\alpha},\frac{n^{\nu-1}}{\xi}}- \mathcal{G}^{(d)}\right| < \dfrac{K_{3,d}(\alpha,\gamma)}{n^{(d+1)/\nu}}, \ \ \  d\geq\nu-1
\end{equation}
for constants $K_{3,d}(\alpha,\gamma)$ depending on $\alpha$ and $\gamma$.
\end{lemma}
\begin{proof}
First we use the convergent series \eqref{eq:zggenseriesdef} to express $z_{g,\nu}$ as a finite sum plus remainder:
 \begin{equation}
    z_{g,\nu}(\gamma) = \sum_{i=\nu-1}^{d-2g\nu} a_{i,g,\nu}(\gamma)n^{-i/ \nu} +R_{g,d}(n,\gamma), \qquad g \le \frac{d-\nu+1}{2\nu} \label{eq:expwithrem}
\end{equation}
with $R_{g,d}(n,\gamma)$ asymptotically bounded by $K_{6,g,d}(\gamma)n^{-(d-2g\nu+1)/\nu}$. We denote the partial sums of the genus expansion as
\begin{equation}
    \mathcal{G}_d := \alpha\sum_{j=0}^{\lfloor \frac{d-\nu+1}{2\nu} \rfloor} \dfrac{z_{j,\nu}(\gamma)}{n^{2j}}.
\end{equation}
The difference between $\mathcal{G}_d$ and $\mathcal{G}^{(d)}$ may be bounded as follows.
\begin{align*}
\vert\mathcal{G}_d- \mathcal{G}^{(d)}\vert & = \alpha \left\vert \sum_{j=0}^{\lfloor \frac{d-\nu+1}{2\nu} \rfloor} \left( \dfrac{1}{n^{2j}} \left( \sum_{i=\nu-1}^{d-2j\nu} a_{i,j,\nu}(\gamma)n^{-i/ \nu} +R_{j,d}(n,\gamma)\right) - \sum_{i=\nu-1}^{d-2j\nu}\dfrac{a_{i,j,\nu}(\gamma)}{n^{2j+i/\nu}} \right) \right \vert \\
& = \alpha\left\vert\sum_{j=0}^{\lfloor \frac{d-\nu+1}{2\nu} \rfloor} \dfrac{1}{n^{2j}} R_{j,d}(n,\gamma) \right\vert < \alpha\sum_{j=0}^{\lfloor \frac{d-\nu+1}{2\nu} \rfloor} K_{6,j,d}(\gamma)n^{-(d+1)/\nu}\\
& < K_{7,d}(\alpha,\gamma)n^{-(d+1)/\nu}.
\end{align*}
A simple application of the triangle inequality brings together the bound above, the estimate provided in equation \eqref{eq:shiftedgenuserrorbound}, and the inequality provided by Lemma \ref{lemma:floorineq}, to complete the proof:
\begin{align*}
   \left|x_{n,\frac{n}{\alpha},\frac{n^{\nu-1}}{\xi}}-  \mathcal{G}^{(d)}\right| &< \left|x_{n,\frac{n}{\alpha},\frac{n^{\nu-1}}{\xi}}- \mathcal{G}_d\right| + \left|\mathcal{G}_d-  \mathcal{G}^{(d)}\right| \\
   &< K_{5,\lfloor \frac{d-\nu+1}{2\nu} \rfloor}(\alpha,\gamma)n^{-(2\lfloor \frac{d-\nu+1}{2\nu} \rfloor+2+\frac{\nu-1}{\nu})} +K_{7,d}(\alpha,\gamma)n^{-(d+1)/\nu} <   K_{3,d}(\alpha,\gamma)n^{-(d+1)/\nu}.
\end{align*}

\end{proof}
This enables us to define a new asymptotic expansion for $x_n$ in the gauge $n^{-1/\nu}$, derived from the genus expansion, which we write
\begin{equation}\label{eq:adaptedgenusexpansion}
    x_{n,\frac{n}{\alpha},\frac{n^{\nu-1}}{\xi}}\sim \alpha\sum_{j=0}^{\infty}\sum_{i=\nu-1}^{\infty}\dfrac{a_{i,j,\nu}(\gamma)}{n^{2j+i/\nu}}.
\end{equation}

\section{Center Manifold Expansion with $\alpha$ and $\xi$ fixed}\label{sec:adapcent}
In what follows, for consistency with the main text, we use the phrase ``center manifold expansion" to refer to the asymptotic expansion of the Freud orbit as $n \to \infty$. It should be noted that although the connection to a center manifold has only been established when $\nu =2$ \cite{bib:elt22}, the existence of an asymptotic expansion for $\nu > 2$ is known from the work of M\'at\'e, Nevai, and Zaslavsky \cite{bib:mnz85}. To extend the discussion to the case of $2\nu$-valent maps, we start from the Freud equation \eqref{eq:generalfreudrec} instead of dpI. Establishing the validity of the center manifold expansion when $N$ and $r_{2\nu}$ grow with $n$ amounts to showing that its partial sums satisfy a rescaling condition hinted at by the string equation \eqref{eq:zggendef}. Specifically, we will apply the transformation  $N\rightarrow\sigma N, r_{2\nu}\rightarrow\sigma^{\nu-1}r_{2\nu}$.  By letting $n$ play the role of $\sigma$, we will see that the parameters $N$ and $r_{2\nu}$ can be made to go to infinity with $n$, while keeping control of the error term. This occurs because the error bounds for the rescaled expansion can be related back to the error bounds for $N$ and $r_{2\nu}$ finite by scaling out the asymptotic variable $n$. 

Take the general even weight of the form 
\begin{equation}\label{eq:generalweight}
    w(\lambda) = \exp\left[-N\left({\dfrac{\lambda^2}{2}+\dfrac{r_{2\nu}\lambda^{2\nu}} {2\nu}}\right)\right],
\end{equation}
where $\nu>1$ is a positive integer. Freud's equation \cite{bib:fre76} in this context, also referred to as the discrete string equation \cite{bib:ew}, gives:
\begin{equation}\label{eq:generalfreudrec}
    n = b_n N(J+r_{2\nu}J^{2\nu-1})_{n,n-1},
\end{equation}
where the subscript is the $(n,n-1)$ entry of the matrix sum $J+r_{2\nu}J^{2\nu-1}$ (starting row/column indexing at 0), and $J$ is the semi-infinite Jacobi matrix 

\begin{equation}
    \begin{pmatrix}
    0 &  b_1&0 & 0 &0&0 &0 \hdots \\
    b_1 & 0& b_2 & 0 &0&0 &0 \hdots \\
     0&b_2&0&b_3 & 0 &0 &0\hdots \\
    0  &0  &b_3&0&b_4  &0 &0\hdots \\
    0  &0  &0&\ddots& \ddots &\ddots  &0 \hdots \\
    \end{pmatrix}.
\end{equation}
Note that this matrix simply encodes the recurrence \eqref{eq:rec2}.
Expressed in terms of $x_n$, we find that \eqref{eq:generalfreudrec} gives
\begin{equation}\label{eq:gennurec}
    \dfrac{n}{N}= x_n + r_{2\nu}M_{\nu},
\end{equation}
where
\begin{equation} \label{eq:M_nu}
M_\nu =  M_{\nu}(x_{n+j}:|j|< \nu) = \sum_P \prod_{m=1}^\nu b^2_{n + \ell_m(P)}
= \sum_P \prod_{m=1}^\nu x_{n + \ell_m(P)}
\end{equation}
so that $M_\nu$ is a polynomial of degree $\nu$ in $x_{n+j}$ for $|j| < \nu$.
The sum here runs over a set of planar lattice walks, $P$, known as {\it Dyck paths}, which start at height $n$ and terminate at height $n-1$ and are of length 
$2\nu - 1$. $\ell_{m}(P)$ denotes the deviation of the path $P$ from height $n$ at step $m$. This representation implies that $|\ell_{m}(P)| < \nu$. See 
\cite{bib:er14} for more details on this combinatorial interpretation.

 The key feature of $M_{\nu}$ for us is the useful rescaling condition: 
\begin{equation}
    M_{\nu}(\sigma x_{n+j}:|j|< \nu)= \sigma^{\nu}M_{\nu}(x_{n+j}:|j|< \nu).
\end{equation}
 With this known structure of $M_{\nu}$, let us denote the center manifold expansion, for general $\nu$, as 
\begin{equation}\label{eq:generalnucenter}
    x_n \sim \sum_{i=-1}^\infty c_i n^{-\frac{i}{\nu}}=:\sum_{i=-1}^\infty c_i( N, r_{2\nu}) n^{-\frac{i}{\nu}}.
\end{equation}
As mentioned above, its existence is known from \cite{bib:mnz85}.

\begin{lemma}\label{lemma:appnrdegreesub} The coefficients $c_i(N,r_{2\nu})$ of the center manifold expansion \eqref{eq:generalnucenter} satisfy the rescaling condition
\begin{equation}
    c_i( \sigma N, \sigma^{\nu-1}r_{2\nu}) = \dfrac{1}{\sigma}c_i(N,r_{2\nu}).
\end{equation}
\end{lemma}

\begin{proof} This result is proved by strong induction.

\noindent\underline{Base case}:
Considering the dominant balance of the equation \eqref{eq:gennurec}, we find that $c_{-1}$ is defined by the equation
\begin{equation}
\dfrac{ 1}{N}= r_{2\nu}c_{-1}^\nu M_{\nu}(\Vec{1})
\end{equation}
where the notation $\Vec{1}$ means that all the coefficients $x_{n+j}$ in \eqref{eq:M_nu} have been set equal to 1; in other words, $M_{\nu}(\Vec{1})$ simply counts unweighted paths.
Thus
\[
\frac{1}{\sigma N}=\big(\sigma^{\nu-1} r_{2\nu}\big) \left(\frac{c_{-1}}{\sigma}\right)^\nu M_{\nu}(\Vec{1}),
\]
indicating that $c_{-1}$ satisfies the desired scaling.

\noindent\underline{Inductive step}:
Assume this rescaling holds for all $c_i$ with $i<m$ and take $m>-1$. The defining equation for $c_m$ is derived from satisfying equation \eqref{eq:gennurec} at order $n^{1-\frac{m+1}{\nu}}$ , when $x_n$ is substituted with 
\begin{equation}
    \sum_{i=-1}^m c_i n^{-\frac{i}{\nu}}.
\end{equation}
At this order we find
\begin{equation}
    0=c_{m+1-\nu}+ r_{2\nu}[c_{m}M(\Vec{1})c_{-1}^{\nu-1}+R(c_j)] \Longleftrightarrow c_m = -\dfrac{c_{m+1-\nu}}{r_{2\nu}M_{\nu}(\Vec{1})c_{-1}^{\nu-1}}-\dfrac{R(c_j)}{M_{\nu}(\Vec{1})c_{-1}^{\nu-1}}
\end{equation}
where we can define $c_{j}=0,$ for  $j<-1$, and $R(c_j)$ is a homogeneous polynomial of degree $\nu$ in $c_j$ for $j<m$  (the remaining terms from $M_{\nu}$ at order $n^{1-\frac{m+1}{\nu}}$ which did not contain $c_m$).

Thus by the homogeneity of $R$ and the inductive hypothesis, we have that the lemma follows.
\end{proof}

\begin{cor}\label{cor:cmpsrescale}
Partial sums of the center manifold expansion \eqref{eq:generalnucenter} satisfy the rescaling 

\begin{equation*}
   \sum_{i=-1}^m c_i(\sigma N, \sigma^{\nu-1}r_{2\nu}) n^{-\frac{i}{\nu}}  = \dfrac{1}{\sigma} \sum_{i=-1}^m c_i(N, r_{2\nu}) n^{-\frac{i}{\nu}}. 
\end{equation*}
\end{cor}

\subsection{ Freud Orbit Rescaling}\label{sec:generalnuproof}
Let us write $x_{n,N,r_{2\nu}}=:x_n(N,r_{2\nu})$ to emphasize dependence on parameters $N$ and $r_{2 \nu}$.

\begin{theorem}\label{thm:orbitrescale}
The $x_{n,N,r_{2\nu}}$ satisfy the following rescaling:
\begin{equation}
    x_{n,\sigma N,\sigma^{\nu-1}r_{2\nu}}=\dfrac{1}{\sigma}x_{n,N,r_{2\nu}}.
\end{equation}
\end{theorem}
\begin{proof}
Denote the $i$\textsuperscript{th} moment for the weight \eqref{eq:generalweight} as $\mu_i=\mu_i(N,r_{2\nu})$. The work of Szeg\"{o} \cite{bib:Sze39} provides explicit formulas for $x_n$ in terms of Hankel determinants, which read
\begin{equation}
x_n = \dfrac{D_{n-2}D_n}{D_{n-1}^2}, \quad \text{where }  D_n=D_n(N,r_{2\nu}): = \begin{vmatrix}
\mu_0& \mu_1&\cdots&\mu_n \\
\mu_1& \mu_2&\cdots&\mu_{n+1} \\
\vdots& \vdots&\vdots&\vdots \\
\mu_{n-1}& \mu_n&\cdots&\mu_{2n-1} \\
\mu_n& \mu_{n+1}&\cdots&\mu_{2n}
    \end{vmatrix}.
\end{equation}
The proof of theorem \ref{thm:orbitrescale} relies on the two following lemmas.
\begin{lemma}
The moments $\mu_n(N,r_{2\nu})$ satisfy the rescaling relation:
\begin{equation}
  \mu_n( \sigma N,\sigma^{\nu-1}r_{2\nu})= \sigma^{-(n+1)/2}\mu_n(N,r_{2\nu}).
\end{equation}
\end{lemma}

\begin{proof}
The proof of the lemma follows from a straightforward change of variables.
\begin{align*}
    \mu_n(\sigma N,\sigma^{\nu-1}r_{2\nu}) &= \int_\mathbb{R} \lambda^n \exp\left[-\sigma N\left({\dfrac{\lambda^2}{2}+\dfrac{\sigma^{\nu-1}r_{2\nu}\lambda^{2\nu}} {2\nu}}\right)\right] d\lambda \\
    &= \int_\mathbb{R} \left(\dfrac{\theta}{\sqrt{\sigma}}\right)^n \exp\left[-\sigma N\left({\dfrac{\left(\dfrac{\theta}{\sqrt{\sigma}}\right)^2}{2}+\dfrac{\sigma^{\nu-1}r_{2\nu}\left(\dfrac{\theta}{\sqrt{\sigma}}\right)^{2\nu}} {2\nu}}\right)\right] d\left(\dfrac{\theta}{\sqrt{\sigma}}\right)\\
    &= \sigma^{-(n+1)/2}\int_\mathbb{R} \theta^n \exp\left[- N\left({\dfrac{\theta^2}{2}+\dfrac{r_{2\nu}\theta^{2\nu}} {2\nu}}\right)\right] d\theta \\
    &= \sigma^{-(n+1)/2}\mu_n(N,r_{2\nu}).
\end{align*}
\end{proof}

\begin{lemma}
The Hankel determinants $D_n$ satisfy the rescaling relation:
\begin{equation}
   D_n(\sigma N,\sigma^{\nu-1}r_{2\nu})=\sigma^{\frac{-(n+1)^2}{2}}D_n(N,r_{2\nu}).
\end{equation}
\end{lemma}

\begin{proof}
\begin{align*}
    D_n(\sigma N,\sigma^{\nu-1}r_{2\nu}) &= \sum_{\rho \in S_{n+1}} sgn(\rho) \prod_{i=1}^{n+1} \mu_{i+\rho(i)-2}(\sigma N,\sigma^{\nu-1}r_{2\nu}) \\
    &= \sum_{\rho \in S_{n+1}} sgn(\rho) \prod_{i=1}^{n+1} \sigma^{\frac{-(i+\rho(i)-1)}{2}} \mu_{i+\rho(i)-2}(N,r_{2\nu}) \\
    &= \sigma^{\frac{-(n+1)(n+2)+(n+1)}{2}}D_n(N,r_{2\nu})\\
    &= \sigma^{\frac{-(n+1)^2}{2}}D_n(N,r_{2\nu}).
    \end{align*}
\end{proof}   
\noindent Finally we directly deduce theorem \ref{thm:orbitrescale}:
\begin{align*}
        x_{n,\sigma N,\sigma^{\nu-1}r_{2\nu}} &= \dfrac{D_{n-2}(\sigma N,\sigma^{\nu-1}r_{2\nu})D_n(\sigma N,\sigma^{\nu-1}r_{2\nu})}{D_{n-1}(\sigma N,\sigma^{\nu-1}r_{2\nu})^2}\\
        &=\dfrac{\sigma^{\frac{-(n-1)^2}{2}}D_{n-2}(N,r_{2\nu})\sigma^{\frac{-(n+1)^2}{2}}D_n(N,r_{2\nu})}{\sigma^{-n^2}D_{n-1}(N,r_{2\nu})^2}\\
        &=\dfrac{1}{\sigma}\dfrac{D_{n-2}(N,r_{2\nu})D_n(N,r_{2\nu})}{D_{n-1}(N,r_{2\nu})^2}\\
        &=\dfrac{1}{\sigma}x_{n,N,r_{2\nu}}.
\end{align*}
\end{proof} 

\section{Comparison of Expansions}\label{sec:compare}
Let $r_{2\nu}={n^{\nu-1}}/{\xi}$ and $N={n}/{\alpha}$. We can now let $r_{2\nu}$ and $N$ go to infinity with $n$ at these relative rates. First, we derive an error bound for the rescaled center manifold expansion, now that it is clear how its partial sums and the $x_{n,N,r_{2\nu}}$ behave under the above rescaling.

\subsection{Equivalence of the genus and center manifold expansions}
\begin{lemma}\label{lemma:centeradaptbound}
With $r_{2\nu}$ and $N$ related to $n$ as above, as $n$ goes to infinity, the partial sums of the center manifold expansion can be rescaled with $x_n$ to derive the following error bound:
\begin{equation}
   \left|x_{n,\frac{n}{\alpha},\frac{n^{\nu-1}}{\xi}}- \sum_{k=-1}^m \dfrac{c_{k}\left(\frac{1}{\alpha},\frac{1}{\xi}\right)}{n^{1+k/\nu}} \right|
< \dfrac{K_{2,m}(1/\alpha,1/\xi)}{n^{(m+1+\nu)/\nu}}, \ \ \  m\geq-1
\end{equation}
where the constant $K_{2,m}$ from the center manifold estimate \eqref{eq:k2} only depends on $\alpha$ and $\xi$.
\end{lemma}

\begin{proof}
Recall that the center manifold expansion is defined in Equation \eqref{eq:generalnucenter} as
\begin{equation} \label{eq:cm_gen}
    x_{n,N,r_{2\nu}} \sim \sum_{i=-1}^\infty c_i( N, r_{2\nu}) n^{-\frac{i}{\nu}}
\end{equation}
where $N$ and $r_{2 \nu}$ are assumed to be arbitrary but finite.
The proof reduces to using the estimate
\[
\left| x_{n,N,r_{2\nu}} -\sum_{i=-1}^m \dfrac{c_i( N, r_{2\nu})}{n^{i/\nu}} \right| < \dfrac{K_{2,m}(N,r_{2\nu})}{n^{(m+1)/\nu}}, \ \ \  m\geq-1
\]
implied by \eqref{eq:cm_gen} (and corresponding to 
\eqref{eq:k2} of the main text), once we leverage the rescalings established in Corollary \ref{cor:cmpsrescale} and Theorem \ref{thm:orbitrescale}:

\begin{align*}
       \left|x_{n,\frac{n}{\alpha},\frac{n^{\nu-1}}{\xi}}- \sum_{k=-1}^m \dfrac{c_{k}\left(\frac{1}{\alpha},\frac{1}{\xi}\right)}{n^{1+k/\nu}} \right| 
    &= \left|\dfrac{1}{n}x_{n,\frac{1}{\alpha},\frac{1}{\xi}}- \dfrac{1}{n}\sum_{k=-1}^m \dfrac{c_k\left(\frac{1}{\alpha},\frac{1}{\xi}\right)}{n^{k/\nu}} \right|\\
    &< \dfrac{1}{n}\dfrac{K_{2,m}(1/\alpha,1/\xi)}{n^{(m+1)/\nu}}=\dfrac{K_{2,m}(1/\alpha,1/\xi)}{n^{(m+1+\nu)/\nu}}.
\end{align*}
\end{proof}
\begin{theorem}\label{thm:equivalentexpansions}
Let $r_{2\nu}$ and $N$ go to infinity with $n$ at relative rates $r_{2\nu} = {n^{\nu-1}}/{\xi}$ and $N = {n}/{\alpha}$. The genus expansion in the $n^{1/\nu}$ gauge \eqref{eq:adaptedgenusexpansion} and the center manifold \eqref{eq:generalnucenter} are equivalent.
\end{theorem}
\begin{proof}
The proof follows from a simple application of the triangle inequality, together with the bounds established in Lemmas \ref{lemma:genusadaptbound} and \ref{lemma:centeradaptbound}. For $m \ge -1$,
\begin{align*}
\left|\mathcal{G}^{(m+\nu)} -\sum_{k=-1}^m \dfrac{c_k\left(\frac{1}{\alpha},\frac{1}{\xi}\right)}{n^{1+k/\nu}}  \right| &\leq
    \left|x_{n,\frac{n}{\alpha},\frac{n^{\nu-1}}{\xi}}- \mathcal{G}^{(m+\nu)}\right| +  \left|x_{n,\frac{n}{\alpha},\frac{n^{\nu-1}}{\xi}}-   \sum_{k=-1}^m \dfrac{c_k\left(\frac{1}{\alpha},\frac{1}{\xi}\right)}{n^{1+k/\nu}} \right| 
    \\ &<  \dfrac{K_{3,m+\nu}(\alpha,\gamma)}{n^{(m+1+\nu)/\nu}} + \dfrac{K_{2,m}(1/\alpha,1/\xi)}{n^{(m+1+\nu)/\nu}} \\
    &< \dfrac{K_{4,m}(1/\alpha,1/\xi)}{n^{(m+1+\nu)/\nu}},
\end{align*}
where we have used the fact that $\gamma = \alpha^{\nu-1} / \xi$ is independent of $n$.
\end{proof}

\subsection{Finding the $a_{i,g,\nu}$}\label{subsec:extracting}
As noted in  remark \ref{rem:coeffsuffice}, we only need to know the $3g-1$ coefficients $\{a_{i(\nu-1),g,\nu}\}_{i=1}^{3g-1}$ to obtain a closed-form expression for $z_{g,\nu}$ in terms of $z_{0,\nu}$. By Theorem \ref{thm:equivalentexpansions}, it suffices to extract these coefficients from the center manifold expansion, as it is equivalent to the genus expansion in the $n^{1/\nu}$ gauge. Since the $a_{i,g,\nu}$ enter into $\mathcal{G}^{(m+\nu)}$ sequentially, we only need to track the last coefficient needed to solve for $z_{g,\nu}$, which is $a_{(3g-1)(\nu-1),g,\nu}$. Using theorem \ref{thm:equivalentexpansions} and Equation \eqref{eq:adapgenusps} to determine which $c_k$ this equates to, we find
\begin{equation} \label{eq:knucount}
    1+\dfrac{k_\nu}{\nu} = 2g+\dfrac{(3g-1)(\nu-1)}{\nu} \Longleftrightarrow  k_\nu =5g\nu-2\nu-3g+1.
\end{equation}

Returning to the defining equation \eqref{eq:zggendef} for $z_{0,\nu}$, but letting $\gamma^{{1}/{(\nu-1)}} n$ play the role of $n$, we can rewrite the convergent series \eqref{eq:zggenconv} to illustrate the dependence on $\gamma$ explicitly:
\begin{equation*} \label{eq:zggammavariable}
z_{0,\nu} = \sum_{i=\nu-1}^\infty a_{i,0,\nu}(1)\big(\gamma^{\frac{1}{\nu-1}} n\big)^{-i/ \nu},
\end{equation*}
In other words, we find that \begin{equation*}
   a_{i,0,\nu}(\gamma)= a_{i,0,\nu}(1) \, \gamma^{\frac{-i}{\nu(\nu-1)}}.
\end{equation*}
Similarly for higher genus, we can write $z_{g,\nu}$ as a doubly infinite sum
\[
z_{g,\nu} = \sum_{j=0}^\infty \omega_j z_{0,\nu}^j = \sum_{j=0}^\infty \omega_j \left(\sum_{i=\nu-1}^\infty a_{i,0,\nu}(1)\big(\gamma^{\frac{1}{\nu-1}} n\big)^{-i/ \nu}\right)^j,
\] since $z_{g,\nu}$ is a rational function of the convergent series for $z_0$ in powers of $n^{-1/\nu}$.
When collecting terms in powers of $n$ in the above expression, the coefficient of $n^{-k/\nu}$ involves terms in $\gamma^{-1/(\nu (\nu-1))}$ whose exponents add up to $k$. Therefore, we can write
\[
z_{g,\nu} = \sum_{k=\nu-1}^\infty a_{k,g,\nu}(1)\big(\gamma^{\frac{1}{\nu-1}} n\big)^{-k/ \nu}.
\]
Consequently,
\begin{equation*}
   a_{i,g,\nu}(\gamma) =a_{i,g,\nu}(1)\, \gamma^{\frac{-i}{\nu(\nu-1)}}.
\end{equation*} 
Thus, \eqref{eq:adaptedgenusexpansion} may be re-expressed as a bivariate expansion in which $a_{i,g,\nu}$ arises as the unique term which is a multiple of the monomial $\gamma^{\frac{-i}{\nu(\nu-1)}}n^{-(2g+i/\nu)}$. As a result, we can easily find the $a_{i,g,\nu}$ from inspection of order, just as was witnessed when comparing expansions \eqref{eq:s4} and  \eqref{eq:gammanaddress}. In practice, since $\alpha$ is arbitrary but near 1, we may set $\alpha = 1$. Then, $\gamma = \alpha^{\nu -1} / \xi = 1/ \xi$ and since $r_{2\nu} = {n^{\nu-1}}/{\xi}$,
\[
\gamma^{\frac{-i}{\nu(\nu-1)}}n^{-(2g+i/\nu)} = \left(\xi^{-2g} r_{2\nu}^{-(2g + i/\nu)} \right)^{1/(\nu - 1)}.
\]
This shows that the unknown $a_{i,g,\nu}$ may be easily located in the expansion \eqref{eq:adaptedgenusexpansion} first by collecting terms in $\xi^{-2 g / (\nu - 1)}$ and then by identifying the coefficient of $r_{2\nu}^{-(2g + i/\nu)/(\nu-1)}$.

\section{Counts of 4-valent 2-legged maps for genera 0 through 7}
\label{app:z_counts}
Table \ref{table:mapcounttable} below shows counts of unlabeled 2-legged $g$-maps obtained using the Taylor expansion of $z_g$, as described in Equation \eqref{eq:map_counts_z}. As explained in Remark \ref{cartographic} these counts are all integral. The number of labeled 2-legged $g$-maps is obtained by multiplying each row by $4^j\cdot j!$,
where $j$ is the corresponding number of vertices. Considering the Euler characteristic $\chi$ of the cellular polyhedron determined by a $g$-map with $V$ $2\nu$-valent vertices and 2 legs, we see that
\[
\chi = 2 - 2 g = (V + 2) - E + F = (V + 2) - \dfrac{1}{2} \left(2 \nu V + 2\right) + F \ge V(1-\nu) + 2,
\]
where the number of edges is $E = \nu V + 1$, and $F \ge 1$ is the number of faces. The above equation thus implies
\[
V \ge \dfrac{2 g}{\nu -1} \qquad \text{for } \nu > 1. 
\]
Consequently, counts for maps with a number of vertices strictly less than $2 g/(\nu-1)$ are all zero, as observed below in Table \ref{table:mapcounttable} (for which $\nu = 2$). In addition, given that the counts are obtained from Equation \eqref{eq:map_counts_z}, $z_g$, as a function of $z_0$, will have a factor of $(z_0-1)^{\lceil 2 g/(\nu-1)\rceil}$.

\begin{table}[!htbp]
\centering
\begin{tabularx}{.95\textwidth}{ |c| *{3}{Y|} }
\hline
{\bf Vertices}    & {\bf genus 0} & {\bf genus 1} & {\bf genus 2} \\
\hline
1 & $3$ & 0 & 0  \\
\hline
2            & $18$ & $6$ & 0 \\
\hline
3           & $135$ & $162$& 0\\
\hline
4    & $1134$ & $3132$ & $630$ \\
\hline
5    & $10206$  & $52650$ & $37422$ \\
\hline
6    &  $96228$  & $819396$ & $1326780$ \\
\hline
7    & $938223$ & $12145140$ & $36506862$ \\
\hline
8    & $9382230$ & $174067704$& $860304564$ \\ 
\hline
9    & $95698746$ & $2434354074$& $18243857772$ \\
\hline
10    & $991787004$ & $33415041780$& $358304450616$ \\
\hline
11 & 10413763542 & 451988208540 & 6637515628590\\
\hline
12 & 110546105292 & 6041901710664 & 117426287155716\\
\hline
13 & 1184422556700 & 79981821607428 & 2001523611771684\\
\hline
14 & 12791763612360 & 1050193148874408 & 33083648147905992\\
\hline
15 & 139110429284415 & 13694359796856360 & 532922312613419820\\
\hline
\hline
Vertices    & {\bf genus 3} & {\bf genus 4} & {\bf genus 5} \\
\hline
1 - 5 & 0 & 0 & 0 \\
\hline
6 & $207900$ & 0 & 0 \\
\hline
7 & $19943172$ & 0 & 0 \\
\hline
8 & $1061845848$ & $141891750 $ & $0$ \\ 
\hline
9 & $41576155956$ & $19177999830 $ & $0$\\
\hline
10 & $1337625029736$ & $1385054577468$ & $164991726900$  \\
\hline
11 & 37475824661352 & 71327306912598 & 29106185730300 \\
\hline
12 & 946821516450480 & 2942589735251316 & 2681355887787528 \\
\hline
13 & 22071416300654292 & 103495914888426684 & 172697001236536140 \\
\hline
14 & 482336962749597384 & 3224203267738773816 & 8760448586644050744 \\
\hline
15 & 9996484963729255992 & 91261924159660147350 & 373335639088458314520\\
\hline
\end{tabularx}

\smallskip
\begin{tabularx}{.95\textwidth}{ |c| *{2}{Y|} }
\hline
{\bf Vertices} & 
{\bf genus 6} & {\bf genus 7} \\
\hline
1 - 11 & 0 & 0 \\
\hline
12 & 292200348339900 & 0 \\
\hline
13 & 64071279522665100 & 0 \\
\hline
14 & 7226119529305407000 & 732588016195035000\\
\hline
15 & 562103677531247569740 & 193018419151189720200\\
\hline
\end{tabularx}
%}
  \caption{Counts of unlabeled 2-legged 4-valent $g$-maps with a fixed number of vertices, for genera 0 through 7.}
\label{table:mapcounttable}
\end{table}

\section{Counts of 4-valent maps for genera 0 through 7}
\label{app:e_counts}
Table \ref{table:mapcounttable_e} below shows counts of unlabeled $g$-maps obtained using the Taylor expansion of $e_g$ as described in Equation \eqref{eq:map_counts_e}. The number of labeled $g$-maps is obtained by multiplying each row by $4^j\cdot j!$, where $j$ is the corresponding number of vertices. When $g=0$, the resulting count is given by Equation \eqref{eq:zero_map_couts} with $\nu = 2$. In the case of regular maps without legs, we do not expect integral counts before multiplication by $4^j\cdot j!$,
due to the presence of symmetries. As before, if one considers the Euler characteristic $\chi$ of the cellular polyhedron determined by a $2\nu$-valent $g$-map with $V$ vertices, we see that
\[
\chi = 2 - 2 g = V - E + F = V - \nu V + F \ge V(1-\nu) + 1,
\]
where the number of edges is $E = \nu V$, and $F \ge 1$ is the number of faces. The above equation thus implies
\[
V \ge \dfrac{2 g - 1}{\nu -1} \qquad \text{for } \nu > 1. 
\]
Consequently, counts for maps with a number of vertices strictly less than $(2 g - 1)/(\nu-1)$ are all zero, as observed below in Tables \ref{table:mapcounttable_e} (for which $\nu = 2$) and \ref{table:mapcounttable_e3} (for $\nu = 3/2$). This is in accord with Theorem (2.3) of \cite{bib:er14}, which established a conjecture due to \cite{bib:biz}. 

\begin{table}[!htbp]
\centering
%\scalebox{0.92}{
\begin{tabularx}{.95\textwidth}{ |c| *{3}{Y|} }
\hline
{\bf Vertices}    & {\bf genus 0} & {\bf genus 1} & {\bf genus 2} \\
\hline
1 & 1/2 & 1/4 & 0 \\
\hline
2 & 9/8 & 15/8 & 0 \\
\hline
3 & 9/2 & 33/2 & 15/4 \\
\hline
4 & 189/8 & 2511/16 & 2007/16 \\
\hline
5 & 729/5 & 15633/10 & 28323/10 \\
\hline
6 & 8019/8 & 64233/4 & 430029/8 \\
\hline
7 & 104247/14 & 1180251/7 & 1848015/2 \\
\hline
8 & 938223/16 & 57590271/32 & 238356027/16 \\ 
\hline
9 & 483327 & 38914749/2 & 229637187 \\
\hline
10 & 82648917/20 & 850128453/4 & 136971261063/40 \\
\hline
11 & 400529367/11 & 25751800341/11 & 99551516103/2 \\
\hline
12 & 1316025063/4 & 207750029985/8 & 5672523466467/8 \\
\hline
13 & 39480751890/13 & 3767137066053/13 & 9936375583257 \\
\hline
14 & 1598970451545/56 & 45501750431811/14 & 549453974272749/4 \\
\hline
15 & 545531095233/2 & 183072982028274/5 & 1877386504673043 \\
\hline
\hline
{\bf Vertices} & {\bf genus 3} & {\bf genus 4} & {\bf genus 5} \\
\hline
1-4 & 0 & 0 & 0 \\
\hline
5 & 945/2 & 0 & 0 \\
\hline
6 & 125127/4 & 0 & 0 \\
\hline
7 & 8500491/7 & 675675/4 & 0 \\
\hline
8 & 577843065/16 & 555627195/32 & 0 \\ 
\hline
9 & 910934829 & 1967095611/2 & 241215975/2 \\
\hline
10 & 41037618141/2 & 1628891511507/40 & 68510089575/4 \\
\hline
11 & 425429109954 & 2756680837155/2 & 14249112872697/11 \\
\hline
12 & 66226454940987/8 & 323610729315237/8 & 557088690933189/8 \\
\hline
13 & 153195852757365 & 1066627646812359 & 2990111952325347 \\
\hline
14 & 38104924294385091/14 & 206945320458060549/8 & 218346687499327569/2 \\
\hline
15 & 46752178744763622 & 1173050456154224859/2 & 3522319537506492078\\
\hline
\end{tabularx}

\smallskip
\begin{tabularx}{.95\textwidth}{ |c| *{2}{Y|} }
\hline
{\bf Vertices}    & {\bf genus 6} & {\bf genus 7} \\
\hline
1-10 & 0 & 0 \\ \hline
11 & 288735522075/2 & 0 \\ \hline
12 & 211615589730825/8 & 0 \\ \hline
13 & 32850823889930175/13 & 260893168160625 \\ \hline
14 & 4717322888871388995/28 & 117949180927619475/2 \\ \hline
15 & 44155396587351637287/5 & 6851883252610003770\\ \hline
\end{tabularx}
\caption{Counts of unlabeled 4-valent $g$-maps with a fixed number of vertices, for genera 0 through 7.}
\label{table:mapcounttable_e}
\end{table}

\newpage
\section{Counts of 3-valent maps for genera 0 through 2}
\label{app:e_counts3}
Formulas for $e_0$, $e_1$, and $e_2$ as functions of $z_0$ in the case of 3-valent maps were obtained in \cite{bib:ep} and read
\begin{align*}
e_{0}\! \left(t \right) & = 
\frac{\ln \! \left(z_0 \! \left(t \right)\right)}{2}+\frac{\left(z_0 \! \left(t \right)-1\right) \left(z_0 \! \left(t \right)^{2}-6 z_0 \! \left(t \right)-3\right)}{12 z_0 \! \left(t \right)+12}, \\
e_{1}\! \left(t \right) & = 
-\frac{1}{24} \ln \! \left(\frac{3}{2}-\frac{z_0 \left(t \right)^{2}}{2}\right), \\
e_{2}\! \left(t \right) & = 
\frac{\left(z_0 \! \left(t \right)^{2}-1\right)^{3} \left(4 z_0 \! \left(t \right)^{4}-93 z_0 \! \left(t \right)^{2}-261\right)}{960 \left(z_0 \! \left(t \right)^{2}-3\right)^{5}},
\end{align*}
where $z_0(t)$ is implicitly defined by the string equation
\[
1 = z_0 \! \left(t \right)^{2}-72\, t^{2} z_0 \! \left(t \right)^{3}.
\]
Table \ref{table:mapcounttable_e3} below shows counts of unlabeled 3-valent $g$-maps obtained using the Taylor expansion of $e_g$:
\begin{equation}
\label{eq:map_counts_e3}
\text{Count of unlabeled 3-valent $g$-maps with $j$ vertices} = \dfrac{1}{3^j\,j !} \left. \dfrac{d^j e_g}{d t^j}\right\vert_{t=0}.
\end{equation}
The number of labeled $g$-maps is obtained by multiplying each row by $3^j\cdot j!$, where $j$ is the corresponding number of vertices. As before, we do not expect integral counts before multiplication by $3^j\cdot j!$ due to the presence of symmetries. In addition, because $z_0$ is a function of $t^2$ (due to the form of the string equation), odd derivatives of $e_g$, and thus counts of maps with odd numbers of vertices, are all zero. This was to be expected since the edges of the  cellular decomposition provided by the 3-valent map graph arise by a perfect pairing of the $3 j$ half edges (the darts mentioned in Remark \ref{cartographic}) coming from the triplets of edges around the $j$ vertices. This leads to a total edge count of $3 j / 2$. But this edge count can be an integer if and only if $j$ is even. Hence all map counts must be zero when the number of vertices is odd. This feature continues to be true for all regular maps of odd valence. 

\begin{table}[!htbp]
\centering
%\scalebox{0.92}{
\begin{tabular}{|c||c|c|c|}
\hline 
Vertices    & genus 0 & genus 1 & genus 2 \\
\hline
2 & 2/3 & 1/6 & 0 \\
\hline
4 & 8/3 & 7/3 & 0 \\
\hline
6 & 56/3 & 332/9 & 35/6 \\
\hline
8 & 512/3 & 1864/3 & 338 \\ 
\hline
10 & 9152/5 & 54416/5 & 66132/5 \\
\hline
12 & 65536/3 & 1762048/9 & 1305280/3 \\
\hline
14 & 5912192/21 & 25136768/7 & 12963696 \\
\hline
16 & 11534336/3 & 66841600 & 362264064 \\
\hline
18 & 494474240/9 & 33984353024/27 & 29035470208/3 \\
\hline
20 & 12213813248/15 & 358871662592/15 & 1250634104832/5 \\
\hline
22 & 136779182080/11 & 5041100158976/11 & 6301063932672 \\
\hline
24 & 584115552256/3 &  79519344492544/9 & 466648673681408/3 \\
\hline
26 & 40486637895680/13 & 2226722215862272/13 & 3777286156007424 \\
\hline
28 & 355142255771648/7 & 3336406411771904 & 90485142526623744 \\
\hline
30 & 839740501295104 & 978867411892895744/15 & 2142890102656491520 \\
\hline
\end{tabular}
%}
  \caption{Counts of unlabeled 3-valent $g$-maps with a fixed number of vertices, for genera 0 through 2. Only counts for even numbers of vertices are provided since there are no regular odd-valent $g$-maps with an odd number of vertices.}
\label{table:mapcounttable_e3}
\end{table}

The counts in Table \ref{table:mapcounttable_e3} agree with the closed-form expressions given in Equations (1.18) and (1.21) of Bleher and Dea\~no \cite{bib:bd13} for $g=0$ and $g=1$, as well as with the coefficients they provide in Equation (1.29) for $g=2$. Independently, Table \ref{table:mapcounttable_e3} agrees with the formulas for $g=0$ and $g=2$ presented in \cite{bib:ew} (Equations (11-30) and (11-32) respectively)\footnote{~Unfortunately, a term was dropped in the expression of $de_1/dy_0$ appearing in Equation (11-31) of \cite{bib:ew}, leading to incorrect counts being presented for $g=1$ just below that equation (and also below (2-9)). Restoring this omission leads to counts consistent with the numbers shown in Table \ref{table:mapcounttable_e3}, and in agreement with the closed-form expression of \cite{bib:bd13}.} and with the counts shown in Table 1 of \cite{bib:dy17} for $g=0$ through $g=2$, which are provided for up to $j = 12$ vertices.


\newpage
\begin{thebibliography}{}

\bibitem[ACKM93]{bib:a}
J. Ambjorn, L. Chekov, C. F. Kristjansen, and Yu. Makeenko,  {\em Matrix Model Calculations Beyond the Spherical Limit}, Nucl. Phys. B {\bf 404}, 127-172 (1993). 

\bibitem[Ba99]{bib:ball}
P. Ball, {\em The Self-Made Tapestry: Pattern Formation in Nature}, Oxford University Press, 1999.

\bibitem[BD13]{bib:bd13}
P. M. Bleher and A. Dea\~no, {\em Topological Expansion in the Cubic Random Matrix Model}, International Mathematics Research Notices {\bf 2013}, 2699-2755 (2013).

\bibitem[BD16]{bib:bd} 
P. Bleher, and A. Dea\~no, {\em Painlev\'e I double scaling limit in the cubic random matrix model}, Random Matrices: Theory and Applications {\bf 5}, 1650004 (2016).

\bibitem[BGM21]{bib:bgm}
P. Bleher, R. Gharakhloo, K. T-R McLaughlin, {\em Phase Diagram and Topological Expansion in the Complex Quartic Random Matrix Model}, preprint, arXiv:2112.09412v1 (2021).

\bibitem[BGR08]{bib:bgr}
E. A. Bender, Z. C. Gao and L. B. Richmond, {\em The map asymptotics constant $t_g$}, Electronic Journal of Combinatorics {\bf 15}, Research paper 51 (2008). 

\bibitem[BI05]{bib:bi} 
P. Bleher, and A. Its, {\em Asymptotics of the partition function of a random matrix}, Ann. Inst. Fourier (Grenoble) {\bf 55}, 1943-2000 (2005).

\bibitem[BIZ80]{bib:biz}
D. Bessis, C. Itzykson, and J. B. Zuber, {\em Quantum field theory techniques in graphical enumeration}, Advances in Applied Mathematics {\bf 1}, 109-157 (1980).

\bibitem[BK86]{bib:bk86}
E. Brieskorn, H. Knörrer, {\em Plane algebraic curves}, Birkh\"auser, Basel, 1986.

\bibitem[BK90]{bib:bk90}
E. Br\'ezin and V.A. Kazakov, {\em Exactly Solvable Field Theories of Closed Strings}, Phys. Lett. B {\bf 236}, 144-150 (1990).

\bibitem[BMN88]{bib:bmn88}
W. C. Bauldry, A. M\'{a}t\'{e}, and P. Nevai, {\em Asymptotics for Solutions of Systems of Smooth Recurrence Equations}, Pacific Journal of Mathematics {\bf 133}, 209-227 (1988).

\bibitem[Ch09]{bib:chap}
G. Chapuy, {\em Asymptotic enumeration of constellation and related families of maps on orientable surfaces}, Combinatorics, Probability and Computing {\bf 18}, 477-516 (2009).

\bibitem[CMS09]{bib:cms}
G. Chapuy, M. Marcus and  G. Shaeffer, {\em A bijection for rooted maps on orientable surfaces}, SIAM Journal of Discrete Mathematics {\bf 23}, 1587-1611 (2009).

\bibitem[CN06]{bib:cn}
F. Camia and C. M. Newman, {\em Two-dimensional critical percolation: the full scaling limit}, Comm. Math. Phys. {\bf 268}, 1-38 (2006).

\bibitem[DY17]{bib:dy17}
B. Dubrovin, D. Yang, {\em Generating series for GUE correlators}, Lett Math Phys {\bf 107}, 1971–2012 (2017).

\bibitem[ELT22]{bib:elt22}
N. M. Ercolani, J. Lega, and B. Tippings, {\em Dynamics of Nonpolar Solutions to the Discrete Painlev\'e I Equation}, SIAM J. Appl. Dyn. Sys. {\bf 21}, 1322-1351 (2022).

\bibitem[ELT22b]{bib:elt23}
N. M. Ercolani, J. Lega, and B. Tippings, {\em Non-recursive Counts of Graphs on Surfaces}, preprint, arXiv: 2210.00671 (2022).

\bibitem[EM03]{bib:em}
N. M. Ercolani and K. D. T.-R. McLaughlin, {\em Asymptotics of the partition function for random matrices via Riemann  Hilbert techniques, and applications to graphical enumeration},  Int. Math. Res. Not. {\bf 14}, 755-820 (2003).

\bibitem[EMP08]{bib:emp08}
N. M. Ercolani, K. D. T-R McLaughlin, V. U. Pierce, {\em Random Matrices, Graphical Enumeration and the Continuum Limit of Toda Lattices}, Communications in Mathematical Physics {\bf 278}, 31-81 (2008).

\bibitem[EP12]{bib:ep}
N.M. Ercolani and  V. U. Pierce, {\em The Continuum Limit of Toda Lattices for Random Matrices with Odd Weights}, Commun. Math. Sci. {\bf 10}, 267-305 (2012).

\bibitem[Er11]{bib:er}
N. M. Ercolani, {\em Caustics, counting maps and semi-classical asymptotics},  Nonlinearity {\bf 24}, 481-526 (2011).

\bibitem[Er14]{bib:er14}
N. M. Ercolani, {\em Conservation laws of random matrix theory}, Random Matrices {\bf 65}, 163-197 (2014).

\bibitem[EW22]{bib:ew}
N. M. Ercolani and P. Waters, 
{\em Relating random matrix map enumeration to a universal symbol calculus for recurrence operators in terms of Bessel-Appell polynomials}, Random Matrices: Theory and Applications {\bf 11}, 2250037 (2022).

\bibitem[Ey11]{bib:e}
B. Eynard, {\em Formal matrix integrals and combinatorics of maps}, in Random Matrices, Random Processes and Integrable Systems. ed. John Harnad, CRM Series in Mathematical Physics, pp. 415-442, Springer, 2011.

\bibitem[Ey16]{bib:ebook}
B. Eynard, {\em Counting Surfaces: Matrix Models and Algebraic Geometry}, Progress in Mathematical Physics, Vol. 70, Birkh\"auser, 2016.

\bibitem[FIK91]{bib:fik91}
A.S. Fokas, A.R. Its, and A.V. Kitaev, {\em Discrete Painlev\'e Equations and their Appearance in Quantum Gravity}, Commun. Math. Phys. {\bf 142}, 313-344 (1991).

\bibitem[FIK92]{bib:fik}
A. S. Fokas, A. R. Its, and A. V. Kitaev, {\em The isomonodromy approach to matrix models in 2D quantum gravity}, Commun. Math. Phys. {\bf 147}, 395-430 (1992). 

\bibitem[FIKN06]{bib:fikn06}
A. S. Fokas, A. R. Its, A. A. Kapaev, V. Y. Novokshenov, {\em Painlev\'e Transcendents, The Riemann-Hilbert Approach}, AMS
Mathematical Surveys and Monographs Vol. 128, 2006.

\bibitem[FOMG13]{bib:fomg}
A. G. Fletcher, J. M. Osborne, P. K. Maini and D. J. Gavaghan, {\em Implementing vertex dynamics models of cell populations in biology within a consistent computational framework}, Progress in Biophysics and Molecular Biology. {\bf 113}, 299-326 (2013).

\bibitem[Fre76]{bib:fre76}
G. Freud, {\em On the Coefficients in the Recursion Formulae of Orthogonal Polynomials}, Proceedings of the Royal Irish Academy {\bf A 76}, 1-6 (1976).

\bibitem[HFC20]{bib:hfc}
J. Hu, G. Filipuk and Y. Chen, {\em Differential and difference equations for recurrence coefficients of orthogonal polynomials with hypergeometric weights and B\"acklund transformations of the sixth Painlev\'e equation}, Random Matrices: Theory and Applications, 2150029 (2020).
DOI: 10.1142/S2010326321500295.

\bibitem[HS21]{bib:hs}
N. Holden and X. Sun, {\em Convergence of uniform triangulations under the Cardy embedding}, arXiv:1905.13207v3 (2021).

\bibitem[JV90]{bib:jv}
D. M. Jackson and T. I. Visintin, {\em A Character Theoretic Approach to Embeddings of Rooted Maps in an Orientable Surface of Given Genus}, Transactions of the AMS {\bf 327}, 343-363 (1990).

\bibitem[Le11]{bib:l11}
J. Lega, {\em Collective behaviors in two-dimensional systems of interacting particles}, SIAM J. Appl. Dyn. Sys. {\bf 10}, 1213-1231 (2011).

\bibitem[LQ83]{bib:lq83}
J. Lew and D. Quarles, {\em Nonnegative Solutions of a Nonlinear Recurrence},
Journal of Approximation Theory {\bf 38}, 357-379 (1983).

\bibitem[LZ04]{bib:lz}
S. Lando and A. Zvonkin, {\em Graphs on Surfaces and their Applications}, Encyclopedia of Mathematical Sciences, Low Dimensional Topology, Volume 11, 2004.

\bibitem[Map21]{bib:maple}
Maple 2021, Version 2021.2, November 2021, symbolic computing environment by Waterloo Maple (Maplesoft).

\bibitem[Mat20]{bib:mathematica}
Mathematica 12, Version 12.2, December 2020, symbolic computing environment by Wolfram.

\bibitem[MB14]{bib:mb}
P. K. Maini and R. E. Baker, {\em Modelling Collective Cell Motion in Biology}, in Advances in Applied Mathematics, 87, A.R. Ansari (ed), Springer International Publishing, 2014.

\bibitem[MNZ85]{bib:mnz85}
A. M\'at\'e, P. Nevai, T. Zaslavsky, {\em Asymptotic expansions of ratios of coefficients of orthogonal polynomials with exponential weights}, Trans. Amer. Math. Soc. {\bf 287}, 495-505 (1985).

\bibitem[Pi06]{bib:vp}
V. Pierce. {\em An Algorithm for Map Enumeration}, arXiv:math/0610586 (2006). DOI: \url{https://doi.org/10.48550/arXiv.math/0610586}. Code available on GitHub at \url{https://github.com/virgilpierce/Vertex_Counting}.

\bibitem[Sz39]{bib:Sze39}
G. Szeg\"{o}, {\em Orthogonal Polynomials},
American Mathematical Society, 1939.

\bibitem[Tip20]{bib:tip20}
B. Tippings, \textit{Discrete Painlev\'e Equations, Orthogonal Polynomials, and Counting Maps,} PhD Dissertation, The University of Arizona, 2020.

\bibitem[Tu68]{bib:tu}
W. T. Tutte, {\em On the enumeration of planar maps}, Bull. Amer. Math. Soc. {\bf 74}, 64-74 (1968).

\bibitem[W15]{bib:w}
P. Waters, {\em Solution of String Equations for Asymmetric Potentials}, Nucl. Phys. B {\bf 899}, 265-288 (2015).



\end{thebibliography}
\end{document}